\documentclass[fleqn,11pt,twoside]{article}

\usepackage{amsthm,amsthm,amssymb, color, xcolor,epsfig, graphics, subfigure}

\usepackage{amsmath, graphicx, latexsym, lscape}

\usepackage{tikz,pgf}
\usetikzlibrary{arrows}
\newcommand{\nn}{\nonumber}

\newcommand{\be}{\begin{equation}}
\newcommand{\ee}{\end{equation}}
\newcommand{\bea}{\begin{eqnarray}}
\newcommand{\eea}{\end{eqnarray}}
\newcommand{\bse}{\begin{subequations}}
\newcommand{\ese}{\end{subequations}}

\newcommand{\vf}{\varphi}

\DeclareMathAccent{\wtilde}{\mathord}{largesymbols}{"65}
\DeclareMathAccent{\what}{\mathord}{largesymbols}{"62}

\def\m@th{\mathsurround=0pt}
\mathchardef\bracell="0365
\def\upbrall{$\m@th\bracell$}
\def\undertilde#1{\mathop{\vtop{\ialign{##\crcr
    $\hfil\displaystyle{#1}\hfil$\crcr
     \noalign
     {\kern1.5pt\nointerlineskip}
     \upbrall\crcr\noalign{\kern1pt
   }}}}\limits}

\newcommand{\wb}[1]{\overline{#1}}
\newcommand{\ub}[1]{\underline{#1}}
\newcommand{\wh}{\widehat}
\newcommand{\wt}{\widetilde}
\newcommand{\ut}{\undertilde}

\newcommand{\uh}{\underhat} 
\def\hypohat#1#2{\vrule depth #1 pt width 0pt{\smash{{\mathop{#2}
\limits_{\displaystyle\widehat{}}}}}}

\usepackage{mathrsfs}
\usepackage{amsmath,amssymb,amsthm}
\usepackage{cite}
\usepackage{color}
\usepackage{hyperref}
\usepackage{tikz,pgf}
\usetikzlibrary{arrows}
\usepackage{cancel}
\usepackage[mathscr]{euscript}  
\usepackage{mathtools} 
\numberwithin{equation}{section} 

\usepackage{amsmath,amssymb}
\usepackage[margin=3cm]{geometry}
\usepackage{hyperref}
\usepackage{tikz}

\usepackage{pgfpages}
\usepackage{tikz}
\usepackage[utf8x]{inputenc}

\usetikzlibrary{calc}
\newtheorem{theorem}{\bf Theorem}[section]

\numberwithin{equation}{section}

\def\undertilde#1{\mathord{\vtop{\ialign{##\crcr
$\hfil\displaystyle{#1}\hfil$\crcr\noalign{\kern1.5pt\nointerlineskip}
$\hfil\widetilde{}\hfil$\crcr\noalign{\kern-6.5pt}}}}}
\def\underhat#1{\mathord{\vtop{\ialign{##\crcr
$\hfil\displaystyle{#1}\hfil$\crcr\noalign{\kern1.5pt\nointerlineskip}
$\hfil\widehat{}\hfil$\crcr\noalign{\kern-6.5pt}}}}}
\def\underbar#1{\mathord{\vtop{\ialign{##\crcr
$\hfil\displaystyle{#1}\hfil$\crcr\noalign{\kern1.5pt\nointerlineskip}
$\hfil\bar{}\hfil$\crcr\noalign{\kern-6.5pt}}}}}

\def\wt#1{\widetilde{#1}}
\def\wh#1{\widehat{#1}}
\def\wb#1{\bar{#1}}

\def\ut#1{\undertilde{#1}}
\def\uh#1{\underhat{#1}}
\def\ub#1{\underbar{#1}}

\makeatletter
\newcommand{\copyrightnote}[2]{{\renewcommand{\thefootnote}{}
 \footnotetext{\small\it
\begin{flushleft}
 \copyright \ #1   #2  
\end{flushleft}}}}

\newcommand{\Name}[1]{\begin{flushleft}
                       \LARGE \bf #1
                       \end{flushleft}\vspace{-3mm}}

\newcommand{\Author}[1]{\begin{flushleft}
                       \it #1 \end{flushleft}}

\newcommand{\Address}[1]{\begin{flushleft}
                       \it #1 \end{flushleft}}

\newcommand{\Date}[1]{\begin{flushleft}
                      \small  \it #1 \end{flushleft}}

\newcommand{\evenhead}{Author \ name}
\newcommand{\oddhead}{Article \ name}

\renewcommand{\@evenhead}{
\hspace*{-3pt}\raisebox{-15pt}[\headheight][0pt]{\vbox{\hbox to \textwidth
{\thepage \hfil \evenhead}\vskip4pt \hrule}}}
\renewcommand{\@oddhead}{
\hspace*{-3pt}\raisebox{-15pt}[\headheight][0pt]{\vbox{\hbox to \textwidth
{\oddhead \hfil \thepage}\vskip4pt\hrule}}}
\renewcommand{\@evenfoot}{}
\renewcommand{\@oddfoot}{}

\long\def\@makecaption#1#2{%
  \vskip\abovecaptionskip
  \sbox\@tempboxa{\small \textbf{#1.}\ \ #2}%
  \ifdim \wd\@tempboxa >\hsize
    {\small \textbf{#1.}\ \ #2}\par
  \else
    \global \@minipagefalse
    \hb@xt@\hsize{\hfil\box\@tempboxa\hfil}%
  \fi
  \vskip\belowcaptionskip}

\newcommand{\JNMPnumberwithin}[3][\arabic]{%
  \@ifundefined{c@#2}{\@nocounterr{#2}}{%
    \@ifundefined{c@#3}{\@nocnterr{#3}}{%
      \@addtoreset{#2}{#3}%
      \@xp\xdef\csname the#2\endcsname{%
        \@xp\@nx\csname the#3\endcsname .\@nx#1{#2}}}}%
}

\renewenvironment{proof}[1][\proofname]{\par
  \normalfont
  \topsep6\p@\@plus6\p@ \trivlist
  \item[\hskip\labelsep\textbf{%
    #1\@addpunct{.}}]\ignorespaces
}{%
  \qed\endtrivlist
}

\newcommand{\resetfootnoterule} {
  \renewcommand\footnoterule{%
  \kern-3\p@
  \hrule\@width.4\columnwidth
  \kern2.6\p@}
}

\renewcommand{\footnoterule}{}

\makeatother

\theoremstyle{definition}

\begin{document}

\renewcommand{\evenhead}{ {\LARGE\textcolor{blue!10!black!40!green}{{\sf \ \ \ ]ocnmp[}}}\strut\hfill 
F W Nijhoff
}
\renewcommand{\oddhead}{ {\LARGE\textcolor{blue!10!black!40!green}{{\sf ]ocnmp[}}}\ \ \ \ \  
Lagrangian multiforms: discrete and semi-discrete KP systems
}

%%%% Matter for the first page 
\thispagestyle{empty}
\newcommand{\FistPageHead}[3]{
\begin{flushleft}
\raisebox{8mm}[0pt][0pt]
{\footnotesize \sf
\parbox{150mm}{{Open Communications in Nonlinear Mathematical Physics}\ \ \ \ {\LARGE\textcolor{blue!10!black!40!green}{]ocnmp[}}
\quad Special Issue 2, 2024\ \  pp
#2\hfill {\sc #3}}}\vspace{-13mm}
\end{flushleft}}

\FistPageHead{1}{\pageref{firstpage}--\pageref{lastpage}}{ \ \ }

\strut\hfill

\strut\hfill

\copyrightnote{The author(s). Distributed under a Creative Commons Attribution 4.0 International License}

\begin{center}

{\bf {\large Proceedings of the OCNMP-2024 Conference:\\ 

\smallskip

Bad Ems, 23-29 June 2024}}
\end{center}

\smallskip

\Name{Lagrangian multiform structure of discrete and semi-discrete KP systems}

\Author{F.W. Nijhoff$^{\dagger\ddagger}$}

\Address{$\dagger$ School of Mathematics, University of Leeds, Leeds LS2 9JT, UK\\
$\ddagger$ Department of Mathematics, Shanghai University, Shanghai 200444, PR China}

\Date{Received June 21, 2024; Accepted July 1, 2024}

\setcounter{equation}{0}

\begin{abstract}

\noindent 
A variational structure for the potential AKP system is 
established using the novel formalism of a Lagrangian multiforms. 
The structure comprises not only the fully discrete equation on the 
3D lattice, but also its semi-discrete variants including several 
differential-difference equations asssociated with, and 
compatible with, the partial difference equation. To this end,    
an overview is given of the various (discrete and semi-discrete) 
variants of the KP system, and their associated Lax representations, 
including a novel `generating PDE' for the KP hierarchy. 
The exterior derivative of the Lagrangian 3-form for the lattice potential KP 
equation is shown to exhibit a double-zero structure, which implies the 
corresponding generalised Euler-Lagrange equations.  
Alongside the 3-form structures, we develop a variational formulation of 
the corresponding Lax systems via the square eigenfunction representation 
arising from the relevant direct linearization scheme. 
\end{abstract}

\label{firstpage}

%%%% The Article text starts here

\section{Introduction}%\label{S:Intro}

The notion of Lagrangian multiforms was introduced in \cite{LobbNij2009} to provide a variational formalism 
for systems integrable in the sense of multidimensional consistency (MDC). 
Thus, the multiform theory distinguishes itself from conventional variational approaches  
by the feature that the corresponding Euler-Lagrange equations produce not just a single 
equation per component of the field variable, but a compatible system of equations on the same dependent variable in a multidimensional space of independent variables.  
The corresponding action is a functional of not only the field variables, but also of 
the $d$-dimensional hypersurfaces over which a Lagrangian $d$-form is integrated 
in an ambient space of arbitrary dimension. 

Initially set up for integrable systems in 1+1 dimensions (Lagrangian 2-forms) and 1+0 dimensions 
(Lagrangian 1-forms, cf. \cite{Y-KLN}), the first example of a 2+1-dimensional/3-dimensional case   
(Lagrangian 3-forms) was established in the fully discrete case in  \cite{LNQ}, and in the fully continuous 
case in \cite{SNC21}, of the KP system. Semi-discrete KP systems were so far not covered, and the fully 
discrete case was mainly related to the bilinear form of KP due to Hirota, \cite{Hir81}, but not the more natural 
nonlinear (potential) KP form first given in \cite{NCWQ84} 
which reads 
\be\label{eq:potDDDKP} 
(p-q+\wh{u}-\wt{u})(r+\wh{\wt{u}}) +(q-r+\wb{u}-\wh{u})(p+\wh{\wb{u}}) 
+ (r-p+\wt{u}-\wb{u})(q+\wt{\wb{u}})=0\ . 
\ee
Here the dependent variable $u=u(n,m,h)$ depends on discrete lattice 
variables $n,m,h$ labelling a three-dimensional lattice, and $p,q,r$ 
are associated lattice parameters which one can associate with the 
links on the lattice in $n$,$m$, and $h$-directions respectively. 
The accents $\wt{\phantom{a}}$~,~$\wh{\phantom{a}}$~,~$\wb{\phantom{a}}$~ 
denote elementary lattice shifts in the three lattice directions, i.e. 
$\wt{u}=T_pu=u(n+1,m,h)$, $\wh{u}=T_qu=u(n,m+1,h)$, $\wb{u}=T_ru=u(n,m,h+1)$ with shift operators $T_p,T_q,T_r$ in the $n$, $m$, and $h$-directions\footnote{The notation of the shift operator 
$T_p$ as acting on the variable $n$, etc., should not be confused with a 
shift in the parameter $p$.}. 
Eq. \eqref{eq:potDDDKP} can be shown to 
lead to the potential KP equation in a specific full continuum limit, 
yielding 
\be\label{eq:potKP} 
\partial_x\left(u_t-\tfrac{1}{4}u_{xxx}-\tfrac{3}{2}u_x^2\right)= 
\tfrac{3}{4}u_{yy}\ ,  
\ee 
for the same function $u$ but in terms of appropriate continuous 
variables $x,y,t$. 

It is well-known, from its construction in \cite{NCWQ84} as well as through the classification problem 
addressed in \cite{ABS12}, that \eqref{eq:potDDDKP} is multidimensionally 
consistent, which implies that the lattice can be extended to a multidimensional 
lattice of arbitrary dimension in which on each elementary cube the equation 
can be imposed and that all equations on all faces of elementary hypercubes are consistent. 
It is this fundamental integrability property that characterises many 
integrable lattice systems, and that we aim at capturing in the variational formalism 
of Lagrangian multiforms.

In this note I will first give a brief review of the discrete, semi-discrete 
KP systems associated with \eqref{eq:potDDDKP}, altogether forming a 
large mixed MDC integrable system, and I will present the corresponding 
linear systems (i.e., Lax representations). I will also summarise the 
direct linearising transform approach, cf. \cite{NCWQ84,FN16}, as it 
is linked to the `square eigenfunction' expansions that is needed 
for the variational description of the Lax pair, noting that 
the multiform description as a byproduct also yields a variational 
formulation of not only the nonlinear equation under consideration, but 
also of the Lax pairs, cf. \cite{SNC20}. There are several distinct semi-
discrete forms appearing, namely a form with two discrete variables 
and one continuous variable $\xi$, cf. \eqref{eq:DddKP} below, and 
various differential-difference equations with two continuous variables, 
$\xi$ and $\sigma$ or $\tau$, and one discrete variable $n$ or $m$, namely \eqref{eq:DDdKP} below. In addition there are other semi-discrete 
equations with other combinations of (discrete or continuous) independent 
variables, such as \eqref{eq:DDDdKP} and \eqref{eq:uextras}, where 
$v=p-q+\wh{u}-\wt{u}$. Importantly, all these equations are not 
separate formulae, but can be imposed \textit{simultaneously} on the single 
dependent variable function $u=u(n,m,h; \xi,\sigma,\tau)$ because of the 
important property of \textit{multidimensional consistency} that 
makes all these equations mutually compatible.  Thus, one could consider 
the entire system of all these variant equations as the object that 
should be considered the `KP system' , which includes also the standard 
KP hierarchy. The continuous variables $\sigma$, $\tau$ can be identified 
with the so-called Miwa variables, \cite{Miwa82}, by means of 
\be\label{eq:Miwa} \partial_\tau u=\sum_{j=0}^\infty \frac{1}{p^{j+1}}\frac{\partial}{\partial t_{j}}\ , \quad 
\partial_\sigma u=\sum_{j=0}^\infty \frac{1}{q^{j+1}}\frac{\partial}{\partial t_{j}}\ ,
  \ee 
where the $t_j$, $j=1,2, \cdots$ are the usual time-variables of the KP hierarchy. As a novel result, a fully continuous coupled PDE system (in terms of $u$ and $v$), \eqref{eq:uvPDE} 
together with \eqref{eq:digamma}, is presented with independent variables $\xi,\tau,\sigma$ 
which encodes the entire KP hierarchy. This is what we would call a 
\textit{generating PDE} for the hierarchy. This system allows us to circumvent the rather laborious representation of the KP hierarchy through pseudo-differential operators, since it is closed-form PDE 
system, and only requires the expansions \eqref{eq:Miwa} to unpeel 
the equation in order to obtain the KP hierarchy in terms of the usual time-variables.   
 
The outline of the paper is as follows. In section 2, we summarise the KP system and present the corresponding linear problems in terms of 
auxiliary functions $\vf$ and $\psi$ (the latter solving an adjoint 
Lax pair). In section 3 we present the direct linearising transform 
and the corresponding quadratic eigenfunction expansion. In section 4 
we give the Lagrangian multiform for the semidiscrete KP, and derive 
from it the one for the fully discrete case. In particular, we show 
that the latter possesses the double zero property, which 
guarantees not only the closure relation, but also the relevant 
Euler-Lagrange system in terms of the so-called corner equations. 
Finally, in the Discussion section we mention some future challenges.  

\section{The discrete, semi-discrete and continuous KP system} 

Apart from \eqref{eq:potDDDKP}, we are interested in the Lagrangian structure of the following 
integrable equation. cf. \cite{NCW85},  
\be\label{eq:DddKP} 
\partial_\xi\ln(p-q+\wh{u}-\wt{u})=\wh{\wt{u}}+u-\wh{u}-\wt{u}\  , 
\ee 
which arises as a so-called straight continuum limit\footnote{For this terminology 
we refer to \cite{HJN16}, Ch. 5.} of \eqref{eq:potDDDKP}. 
Eq. \eqref{eq:DddKP} is a semi-discrete version of the KP equation with two discrete 
independent variables $n,m$ and one continuous independent variable $\xi$. 
As mentioned before, the scalar field $u$ can, in principle, be considered to be a function depending on an arbitrary number of discrete variables (with the notation introduced 
in section 1) and possesses, like \eqref{eq:potDDDKP}, the MDC property. 
This equation was used in \cite{NijPeng} to derive a discrete-time version of the 
Calogero-Moser system by pole reduction, and was considered also in \cite{Y-KLN} 
as a starting point for the first occurrence of a Lagrangian 1-form structure, namely for both the 
discrete- and continuous-time Calogero-Moser (CM) systems, but an initial obstacle in that paper was 
that, while we established a variational description of both the discrete and continuous 
CM hierarchy, we didn't have a Lagrangian description of the unreduced equation, namely 
\eqref{eq:DddKP} itself.  This gap in the treatment we want to address here, alongside the 
establishment of the Lagrange structure for the fully discrete KP system. 

The system defined by \eqref{eq:DddKP}, extended in all lattice directions, is multidimensionally consistent, subject to \eqref{eq:potDDDKP} itself, which for 
latter convenience we can also cast in the equivalent form\footnote{Note that \eqref{eq:dddKP} is just one single 
equation, even though written in two alternative ways.}:  
\be\label{eq:dddKP} 
\frac{p-r+\wh{\wb{u}}-\wh{\wt{u}}}{p-r+\wb{u}-\wt{u}}= \frac{q-r+\wt{\wb{u}}-\wh{\wt{u}}}{q-r+\wb{u}-\wh{u}}
=\frac{p-q+\wh{\wb{u}}-\wt{\wb{u}}}{p-q+\wh{u}-\wt{u}}\  , 
\ee
Identifying \eqref{eq:dddKP} as a discrete potential KP equation, by virtue of its link to 
\eqref{eq:potKP}, a non-potential version of this equation was derived in 
\cite{FN16}, cf. also \cite{Nimmo} for alternative versions of the latter. As mentioned 
before, \eqref{eq:dddKP} is itself 
multidimensionally consistent in terms of consistency on the tessaract, or more precisely on the 
octahedral lattice (cf. \cite{ABS12} for a detailed discussion).  

Both \eqref{eq:dddKP} and \eqref{eq:DddKP} form part of a broader structure, which is revealed through the corresponding linear problems and further extensions.  First, we note 
that eq. \eqref{eq:DddKP} arises as the compatibility condition from the following Lax pair 
\bse \label{eq:vfLax}\begin{align}
& \wt{\vf}=\vf_\xi+(p+u-\wt{u})\vf \label{eq:vfLaxa} \\ 
& \wh{\vf}=\vf_\xi+(q+u-\wh{u})\vf  \label{eq:vfLab}
\end{align}
\ese 
or alternatively from the adjoint Lax pair 
\bse \label{eq:psiLax}\begin{align}
& \psi =-\wt{\psi}_\xi+\wt{\psi}(p+u-\wt{u}) \label{eq:psiLaxa} \\ 
& \psi =-\wh{\psi}_\xi+\wh{\psi}(q+u-\wh{u}) \label{eq:psiLaxb}
\end{align}
\ese 
whereas the fully discrete equation arises from the (inhomogeneous) Lax triplets
\bse \label{eq:ddvfLax}\begin{align}
& \wt{\vf}=(p-\wt{u})\vf+\chi\ , \\ 
& \wh{\vf}=(q-\wh{u})\vf+\chi\ , \\ 
& \wb{\vf}=(r-\wb{u})\vf+\chi\ , 
\end{align}
\ese 
or from the adjoint triplet 
 \bse \label{eq:ddpsiLax}\begin{align}
& \psi =\wt{\psi}(p+u)-\wt{\theta}\ , \\ 
& \psi =\wh{\psi}(q+u)-\wh{\theta}\ , \\ 
& \psi =\wb{\psi}(r+u)-\wb{\theta}\ ,  
\end{align}
\ese 
where $\chi$ and $\theta$ are some auxiliary fields that can be pairwise eliminated from \eqref{eq:ddvfLax} and 
\eqref{eq:ddpsiLax} to yield Lax pairs in a more conventional (homogeneous linear) form.

Furthermore we have the following single-shift semi-continuous KP equations\footnote{Here and in what follows, the under-accents ~$\ut{\cdot}$~ and ~$\uh{\cdot}$~ 
denote the backward shifts to the shifts ~$\wt{\cdot}$~ and ~$\wh{\cdot}$~ respectively, i.e. $\ut{u}=T_p^{-1}u=u(n-1,m,h)$ and $\uh{u}=T_q^{-1}u=u(n,m-1,h)$. Similarly, we also have $\ub{u}=T_r^{-1}u=u(n,m,h-1)$. } 
\bse\label{eq:DDdKP}\begin{align}
&\partial_\xi\ln(1+u_\tau) = \wt{u}+\ut{u}-2u \  , \\ 
&\partial_\xi\ln(1+u_\sigma) = \wh{u}+\uh{u}-2u \  ,
\end{align}
\ese 
which involve the continuous Miwa variables $\tau$ and $\sigma$ of \eqref{eq:Miwa},
which are associated with the lattice parameters 
$p$ and $q$, and hence with the lattice shifts $T_p$ and $T_q$ respectively (and similar equations associated 
with any of the lattice directions in the system). Eqs. \eqref{eq:DDdKP}, which are closely connected to the 2D Toda equation, can be obtained from a `skew continuum limit' 
(again in the terminology of Ch. 5 of \cite{HJN16}) from \eqref{eq:dddKP}, but they 
arise also as the compatibility conditions 
of \eqref{eq:vfLax} together with 
\bse\label{eq:DdvfLax}\begin{align}
&\vf_\tau = -(1+u_\tau)\ut{\vf}\  , \label{eq:DdvfLaxa} \\ 
&\vf_\sigma = -(1+u_\sigma)\uh{\vf}\  , \label{eq:DdvfLaxb}
\end{align}
\ese 
or from the adjoint linear equations 
\bse\label{eq:DdpsiLax}\begin{align}
&\psi_\tau = (1+u_\tau)\wt{\psi}\  , \label{eq:DdpsiLaxa} \\ 
&\psi_\sigma = (1+u_\sigma)\wh{\psi}\  , \label{eq:DdpsiLaxb}
\end{align}
\ese 
together with \eqref{eq:psiLax} respectively. 
From these linear systems we can derive the following purely continuous Lax pairs: 
\bse\label{DDDLax}
\begin{align}
&\vf_{\xi\tau}+(p+u-\wt{u})\vf_\tau+(1+u_\tau)\vf=0 \ , \\ 
&\vf_{\xi\sigma}+(q+u-\wh{u})\vf_\sigma+(1+u_\sigma)\vf=0 \ , \\ 
&v\,\vf_{\sigma\tau}=(1+\wt{u}_\sigma)\vf_\tau-(1+\wh{u}_\tau)\vf_\sigma\  , 
\end{align}\ese
where we have set $v=p-q+\wh{u}-\wt{u}$. The analogous relations for the 
adjoint function $\psi$ are 
\bse\label{DDDLaxadj}
\begin{align}
&\psi_{\xi\tau}-(p+\ut{u}-u)\psi_\tau+(1+u_\tau)\psi=0 \ , \\ 
&\psi_{\xi\sigma}-(q+\uh{u}-u)\psi_\sigma+(1+u_\sigma)\psi=0 \ , \\ 
&\uh{\ut{v}}\,\psi_{\sigma\tau}=(1+\uh{u}_\tau)\psi_\sigma-(1+{\ut{u}}\!\!\!\phantom{a}_\sigma)\psi_\tau\  . 
\end{align}\ese
While \eqref{eq:DDdKP} involve the variable $\xi$, we also have relations 
that only involve the continuous variables $\tau$ and $\sigma$, namely 
\begin{subequations}\label{eq:DDDdKP}\begin{eqnarray}
v\,u_{\sigma\tau}&=&(1+u_\tau)(1+\wt{u}_\sigma)-(1+u_\sigma)(1+\wh{u}_\tau)\  ,  \\ 
\uh{\ut{v}}\,u_{\sigma\tau}&=&(1+u_\sigma)(1+\uh{u}_\tau)-(1+u_\tau)(1+\ut{u}_\sigma)\  , 
\end{eqnarray} \end{subequations}  
which complements \eqref{eq:DDdKP}. 
The following relations on the variable $u$ also hold true
\bse\label{eq:uextras}\begin{align}
&\frac{1+\wh{u}_\tau}{1+u_\tau}= \frac{v}{\ut{v}}\ , \quad 
\frac{1+\wt{u}_\sigma}{1+u_\sigma}= \frac{v}{\uh{v}}\ ,  \label{eq:uextraa}\\  
& \frac{u_{\sigma\tau}}{(1+u_\sigma)(1+u_\tau)} 
= \frac{1}{\uh{v}}-\frac{1}{\ut{v}}\  , \label{eq:uextrab} 
\end{align}
\ese
and which arise from the Lax relations \eqref{eq:ddvfLax} and 
\eqref{eq:DdvfLax}, or equivalently \eqref{eq:ddpsiLax} and 
\eqref{eq:DdpsiLax}. From the latter relations, we can derive a fully 
continuous KP system, i.e. a coupled PDE system for $u$ and $v$ 
in terms of the independent variables $\xi$, $\sigma$ and $\tau$, 
eliminating all the lattice shifts. In fact, from \eqref{eq:uextrab} 
using \eqref{eq:DddKP} we can derive the relation
\begin{align} 
\frac{1}{\uh{v}}+\frac{1}{\ut{v}} = 
\frac{\textstyle 4+ \frac{u_{\sigma\tau}}{(1+u_\sigma)(1+u_\tau)} 
\partial_\xi\ln\left( \frac{u^2_{\sigma\tau}}{(1+u_\sigma)(1+u_\tau)} \right)}{\textstyle 2v+\partial_\xi\ln
\left(\frac{1+u_\tau}{1+u_\sigma}\right)}=: \digamma  \label{eq:digamma} 
\end{align} 
where the expression on the r.h.s., which we call $\digamma$, 
only depends on $v$ and derivatives of $u$, but no longer involving any 
shifts. Combining \eqref{eq:digamma} with \eqref{eq:uextrab} we can 
solve for $\uh{v}$ and $\ut{v}$ in terms of $v$ and $u$ (and its derivatives) 
and then use the relation 
\bse \label{eq:genPDEKP} 
\begin{align}  \label{eq:vvuhPDE}
\partial_\xi\ln\uh{v}=\uh{v}-v+\partial_\xi \ln(1+u_\sigma)\ , 
\end{align} 
or the the alternative relation\footnote{That these two relations are equivalent 
follows from the relation 
\begin{align*} 
\partial_\xi\ln\left(\frac{\ut{v}}{\uh{v}}\right)=2v-\ut{v}-\uh{v}
+\partial_\xi\ln \left(\frac{1+u_\tau}{1+u_\sigma} \right)= 
v+\ut{{\hypohat 3 v}}-\uh{v}-\ut{v}\  ,  
\end{align*} 
which is a consequence of \eqref{eq:DddKP} and where is used that 
\begin{align*} 
 \partial_\xi\ln \left(\frac{1+u_\tau}{1+u_\sigma} \right)=\ut{{\hypohat 3 v}}-v\  , 
\end{align*}
which in turn is a consequence of \eqref{eq:DDdKP}. } 
\begin{align} \label{eq:vvutPDE}
\partial_\xi\ln\ut{v}=-\ut{v}+v+\partial_\xi \ln(1+u_\tau)\ , 
\end{align}\ese 
and insert the expression for $\uh{v}$ or $\ut{v}$ respectively obtained 
from \eqref{eq:digamma} and \eqref{eq:uextrab}, and obtain a PDE 
involving $v$, $v_\xi$ and the derivatives of $u$ up to 
$u_{\xi\xi\sigma\tau}$, which reads 
\bse \label{eq:uvPDE} 
\begin{align} \label{eq:uPDEa} 
v & =\partial_\xi \ln\left((1+u_\sigma)\digamma+\frac{u_{\sigma\tau}}{1+u_\tau}
\right)+\frac{2}{\textstyle \digamma+
\frac{u_{\sigma\tau}}{(1+u_\tau)(1+u_\sigma)}}\ , \\ 
& =- \partial_\xi \ln\left((1+u_\tau)\digamma-\frac{u_{\sigma\tau}}{1+u_\sigma}
\right)+\frac{2}{\textstyle \digamma-
\frac{u_{\sigma\tau}}{(1+u_\tau)(1+u_\sigma)}}\ ,\label{eq:uPDEb} 
\end{align}
To complement this relation, we also have from \eqref{eq:uextraa} the PDE: 
\begin{align} \label{eq:vPDE} 
2v_{\sigma\tau}=& \partial_\sigma\left[ (1+u_\tau)v\left(\digamma-
\frac{u_{\sigma\tau}}{(1+u_\tau)(1+u_\sigma)}\right) \right] \nonumber \\ 
& \quad -\partial_\tau\left[ (1+u_\sigma)v\left(\digamma+
\frac{u_{\sigma\tau}}{(1+u_\tau)(1+u_\sigma)}\right) \right]\  ,   
\end{align} 
\ese 
which together with \eqref{eq:uPDEa} or \eqref{eq:uPDEb}, inserting the expression \eqref{eq:digamma} 
for $\digamma$, forms a coupled system of PDEs for $u$ and 
$v$. If one were to eliminate $v$ from this system, and obtain a single PDE 
in terms of $u$, the highest derivative would be $u_{\xi\xi\sigma\sigma\tau\tau}$. 
We will refer to this coupled system \eqref{eq:uvPDE} of PDEs as the 
\textit{generating PDE of the KP hierarchy}, as by inserting the power series 
expansions \eqref{eq:Miwa} it should contain all equations in the KP hierarchy 
that can be obtained from the Sato scheme, \cite{Sato}.   

\section{$\tau$-function relations} 

The various relations for the function $u$ can be resolved using the $\tau$-function 
$f$. In fact, we have the following relations: 
\bse\label{eq:taurels} 
\begin{align}
& u=\partial_\xi\ln f\  , \\ 
& p-q+\wh{u}-\wt{u}=(p-q)\frac{f\,\wh{\wt{f}}}{\wh{f}\,\wt{f}}\  , \\ 
& 1+u_\tau= \frac{\wt{f}\,\ut{f}}{f^2}\  , \\ 
& 1+u_{\sigma}=\frac{\wh{f}\,\uh{f}}{f^2}\  ,
\end{align} 
\ese 
These identification reduce \eqref{eq:DDDdKP}, \eqref{eq:DDdKP} 
and \eqref{eq:DddKP} to identities
The equations for the $\tau$-function comprise the Hirota-Miwa equation 
\begin{equation}\label{eq:HM}.  
(p-q) \wb{f}\,\wh{\wt{f}} + (q-r) \wt{f}\,\wh{\wb{f}}+(r-p) \wh{f}\,\wt{\wb{f}}=0\ , 
\end{equation} 
as well as the Hirota-type bilinear equations 
\bse\label{eq:HMsgtaurel} 
\begin{align}
& (p-q) \left(f_\sigma \wt{f}-f \wt{f}_\sigma \right)=f\,\wt{f} - \wh{f}\,\uh{\wt{f}}\ , \\  
& (q-p) \left(f_\tau \wh{f}-f \wh{f}_\tau \right)=f\,\wh{f} - \wt{f}\,\ut{\wh{f}}\ , 
\end{align} 
\ese 
which can be ontained from \eqref{eq:HM} by performing a skew continuum limit (in the sense 
of Ch. 5 of \cite{HJN16}). In fact, this can be viewed as the following 
remarkable relation connecting the derivatives with regard to the Miwa variables 
and the lattice shifts on $\tau$-function, namely 
\begin{equation}\label{eq:remarkable}
\partial_\tau f=-\left(T_p^{-1} \frac{\rm d}{{\rm d}p} T_p\right) f\ , \quad 
\partial_\sigma f=-\left(T_q^{-1} \frac{\rm d}{{\rm d}q} T_q\right) f\ . 
\end{equation} 
We also have the relation 
\begin{equation}\label{eq:usigmatau}
(p-q)u_{\sigma\tau}=\frac{\wt{\uh{f}}\,\ut{f}\,\wh{f}-\wh{\ut{f}}\,\wt{f}\,\uh{f}}{f^3}\  . 
\end{equation}
Finally, there is the Toda type differential-difference equation
\begin{equation}\label{eq:MiwaToda} 
(p-q)^2\left( f\,f_{\sigma\tau}-f_\sigma f_\tau\right)= \wh{\ut{f}}\,\wt{\uh{f}}-f^2\  . 
\end{equation} 
Note that a non-autonomous version of \eqref{eq:MiwaToda} was presented in \cite{FN20}. 

A special solution to the KP system, in fact the pole reduction we studied in 
\cite{NijPeng,Y-KLN}, is obtained in terms of the $\tau$-function by setting 
\begin{equation}\label{eq:fpole}
f=\prod_{j=1}^N (\xi-x_j)\quad \Rightarrow\quad u=\sum_{j=1}^N \frac{1}{\xi-x_j}\  , 
\end{equation} 
where the zeroes, respectively the poles, $x_j$ are functions of the discrete variables $n,m,\dots$ as well as of the 
continuous variables $\sigma,\tau,\dots$, but not of $\xi$. By using the Lagrange interpolation formulae 
\bse\label{eq:Lagrinterpol} \begin{align}
&\frac{Y(\xi)}{X(\xi)}=1+ \sum_{j=1}^N \frac{1}{\xi-x_j}\,\frac{Y(x_j)}{X'(x_j)}\  , 
\label{eq:Lagrinterpola}    \\ 
&\frac{Y(\xi)Z(\xi)}{X(\xi)^2}=1+ \sum_{j=1}^N \frac{\partial}{\partial x_j}
\left(\frac{1}{\xi-x_j}\,\frac{Y(x_j)Z(x_j)}{X'(x_j)^2}\right)\  , \label{eq:Lagrinterpolb} \\ 
&\frac{Y(\xi)Z(\xi)W(\xi)}{X(\xi)^3}=1+ \sum_{j=1}^N \frac{1}{2}
\frac{\partial^2}{\partial x_j^2}\left(\frac{1}{\xi-x_j}\,\frac{Y(x_j)Z(x_j)W(x_j)}{X'(x_j)^3}\right) \label{eq:Lagrinterpolc} 
%+ \sum_{j=1}^N \frac{1}{(\xi-x_j)^2}\,
%\frac{\partial}{\partial x_j}\left(\frac{Y(x_j)Z(x_j)W(x_j)}{X'(x_j)^3}\right) \nonumber \\ 
%& \qquad \qquad + \sum_{j=1}^N \frac{1}{(\xi-x_j)^3}\,
%\frac{Y(x_j)Z(x_j)W(x_j)}{X'(x_j)^3}\ , 
\end{align}\ese  
where 
\[ X(\xi)=\prod_{j=1}^N (\xi-x_j)\ , \quad Y(\xi)=\prod_{j=1}^N (\xi-y_j)\ , \quad 
Z(\xi)=\prod_{j=1}^N (\xi-z_j)\ , \quad W(\xi)=\prod_{j=1}^N (\xi-w_j)\  , \] 
where the roots $y_j,z_j,w_j$ do not coincide with any of the roots $x_j$, and the latter are assumed distinct. 
Inserting \eqref{eq:fpole} into the relations \eqref{eq:taurels}
and using the expansions \eqref{eq:Lagrinterpol}, leads to 
\bse\label{eq:CMred}
\begin{align}
& \sum_{j=1}^N \left(\frac{1}{x_i-\wt{x}_j}+\frac{1}{x_i-\ut{x}_j}\right)
-\sum_{j=1\atop j\neq i}^N \frac{2}{x_i-x_j}= 0 \quad \Leftrightarrow\quad 
\frac{\partial}{\partial x_i}\left(\frac{\prod_{j=1}^N (x_i-\wt{x}_j)(x_i-\ut{x}_j)}{\prod_{j\neq i} (x_i-x_j)^2} \right)=0\ , \label{eq:CMreda} \\ 
& \partial_\tau x_i= \frac{\prod_{j=1}^N (x_i-\wt{x}_j)(x_i-\ut{x}_j)}{\prod_{j\neq i} (x_i-x_j)^2}\  , \label{eq:CMredb} \\ 
& \frac{\prod_{j=1}^N (\wh{x}_i-\wt{x}_j)\prod_{j\neq i}(\wh{x}_i-\wh{x}_j)}{\prod_{j=1}^N (\wh{x}_i-x_j)(\wh{x}_i-\wh{\wt{x}}_j)}=p-q\  , \label{eq:CMredc} \\ 
& \textrm{for}\quad i=1,\dots, N, \nonumber 
\end{align}\ese 
(and similar relations with $\wt{x}$ and $\wh{x}$, $p$ and $q$ and $\tau$ and $\sigma$ interchanged), together with 
Eq. \eqref{eq:CMred} are the equations of motion of the discrete-time Calogero-Moser model of \cite{NijPeng,Y-KLN}.  
Furthermore, \eqref{eq:usigmatau} inserting \eqref{eq:fpole} and using 
\eqref{eq:Lagrinterpolc} leads to 
\bse
\begin{align} 
& (p-q) \partial_\sigma\partial_\tau x_i= \frac{\partial}{\partial x_i}\left( 
\frac{\uh{\wt{X}}(x_i) \wh{X}(x_i) \ut{X}(x_i) - \ut{\wh{X}}(x_i) \wt{X}(x_i) 
\uh{X}(x_i)}{X'(x_i)^3}\right) \ , \\ 
&  0= \frac{\partial^2}{\partial x_i^2}\left( 
\frac{\uh{\wt{X}}(x_i) \wh{X}(x_i) \ut{X}(x_i) - \ut{\wh{X}}(x_i) \wt{X}(x_i) 
\uh{X}(x_i)}{X'(x_i)^3}\right)\ . 
\end{align} \ese 
Note that the relation \eqref{eq:CMredc} by Lagrange interpolation also leads to
\bse 
\begin{align}
&1+\sum_{l=1}^N \left(\frac{1}{\wh{x}_j-x_l}\,
\frac {\wt{X}(x_l) \wh{X}(x_l)}{X'(x_l)\wh{\wt{X}}(x_l)}
+\frac{1}{\wh{x}_j-\wh{\wt{x}}_l}\, 
\frac {\wt{X}(\wh{\wt{x}}_l) \wh{X}(\wh{\wt{x}}_l)}{X(\wh{\wt{x}}_l)\wh{\wt{X}}'(\wh{\wt{x}}_l)}
 \right)=0\ , \\ 
& p-q+\sum_{l=1}^N \left(\frac{1}{(\wh{x}_j-x_l)^2}\,
\frac {\wt{X}(x_l) \wh{X}(x_l)}{X'(x_l)\wh{\wt{X}}(x_l)}
+\frac{1}{(\wh{x}_j-\wh{\wt{x}}_l)^2}\, 
\frac {\wt{X}(\wh{\wt{x}}_l) \wh{X}(\wh{\wt{x}}_l)}{X(\wh{\wt{x}}_l)\wh{\wt{X}}'(\wh{\wt{x}}_l)}
 \right)=0\ ,   
\end{align}
\ese
(and similar relations with $\wt{x}$ and $\wh{x}$, $p$ and $q$ interchanged). 
 
\section{Direct Linearising transform} 

\newcommand{\ddint}{\iint_{D} d\zeta(l,l^\prime)}

From here on we will assume that the functions $\vf$ of the Lax pairs can be characterised by a 
(spectral type) parameter $k$, while the adjoint functions can be characterised by a parameter 
$k'$, the dependence on which we will denote by an index, namely $\vf_k$ and $\psi_{k'}$ respectively. These spextal parameters can be identified with the lattice parameters 
asociated with specific directions in the multidimensional lattice, and as such 
the Lax pair functions can be expressed in terms of the $\tau$-function as 
\begin{equation}\label{eq:phipsitau}
\vf_k=\frac{T_{-k}f}{f}\,\vf_k^0\ , \quad \psi_{k'}=\frac{T_{k'}^{-1}f}{f}\,\psi_{k'}^0\  . 
\end{equation}
Here the functions $\vf^0_k$ and $\psi^0_{k'}$ are plane-wave factors of the form 
\bse\label{eq:freeref}\begin{align}
& \vf^0_k= \left[\prod_\nu (p_\nu+k)^{n_\nu}\right]\,\exp\left\{k\xi-\sum_\nu \frac{\tau_\nu}{p_\nu+k}  \right\}\ , \\  
& \psi^0_{k'}= \left[\prod_\nu (p_\nu-k')^{-n_\nu}\right]\,\exp\left\{k'\xi+\sum_\nu \frac{\tau_\nu}{p_\nu-k'}  \right\}\ .
\end{align}
\ese 
which are the solutions of the Lax pair relations when the potential function $u=u^0=0$. In 
\eqref{eq:freeref} the variables $n_\nu$ comprise the discrete variables $n$ and $m$ used 
before, associated with the corresponding lattice parameters $p$, $q$ as specific 
choices for the $p_\nu$, while the corresponding continuous variables $\tau_\nu$ are the 
$\tau$ and $\sigma$ variables in the previous sections.    
  
More generally, the direct linearising transform (DLT), \cite{Pasquier1983,SAF1984,NCWQ84,NC90,PP90}, cf. also 
\cite{FN16}, is a very general abstract dressing scheme that builds new solutions 
from given solutions in terms of linear integral equations of a very general form. I will 
briefly summarise the scheme here. 
  
It starts with the observation that, as a consequence of the linear Lax equations for $\vf$ and $\psi$, have the following Wronskian type (or closure) relations:  
\bse \label{eq:Wronskrels}
\begin{align}
& \Delta_q\left((T_p\psi_{k'})\vf_k \right)=\Delta_p\left((T_q\psi_{k'})\vf_k \right)\ , \\ 
& \partial_\xi\left((T_p\psi_{k'})\vf_k \right)=\Delta_p\left(\psi_{k'}\vf_k \right)\ , \\ 
& \partial_\tau\left((T_p\psi_{k'})\vf_k \right)=\Delta_p\left((1+u_\tau)(T_p\psi_{k'})(T_p^{-1}\vf_k) \right)\ . 
\end{align}
In addition to these there are relations of mixed type, namely
\begin{align}\label{eq:Wronskextrarels}
& \partial_\tau\left((T_q\psi_{k'})\vf_k \right)=\Delta_q\left((1+u_\tau)(T_p\psi_{k'})(T_p^{-1}\vf_k) \right)\ , \\ 
& \partial_\sigma \left((1+u_\tau)(T_p\psi_{k'})(T_p^{-1}\vf_k) \right)=
\partial_\tau\left((1+u_\sigma)(T_q\psi_{k'})(T_q^{-1}\vf_k) \right)
\end{align}
\ese 
which are nontrivial, and hold subject to the relations \eqref{eq:uextras} on the 
variable $u$. 

Obviously similar relations to \eqref{eq:Wronskrels} also hold for any other pairs of shift 
variables and parameters and corresponding continuous time-variables in the system.  These relations suggest that there is 
closed 1-form in the space of discrete as well as continuous variables which we symbolically 
write as 
\begin{equation}\label{eq:1form} 
\Omega_{k,k'}=\psi_{k'}\vf_k\,{\rm d}\xi+\sum_\nu \left[(T_{p_\nu}\psi_{k'})\vf_k\, \delta_{p_\nu} + (1+u_{\tau_\nu})(T_{p_\nu}\psi_{k'})
(T_{p_\nu}^{-1}\vf_k)\,{\rm d}\tau_\nu\right]\  , 
\end{equation} 
where $\nu$ is a symbol that labels all the lattice directions in the system, and $\delta_p$ 
is a symbol denoting that the finite contribution, given by the corresponding coefficient of the 1-form is in the 
direction associated with the parameter $p$.  Closedness of this form \eqref{eq:1form} implies that the 
semisum-integral\footnote{This partly discrete- partly continuous line integral may be understood analytically 
in terms of the theory of timescales, cf. e.g. \cite{BohnPet,Bohner}.}  
\begin{equation}\label{eq:G} 
G_{k,k'}=%{\Huge{\sf{\bf{S}}}}_\Gamma \Omega_{k,k'} 
\mathrlap{\int}\sum_\Gamma \Omega_{k,k'}\ . 
\end{equation}
which combines a sum over discrete path sections as well as a line integral over continuous sections of a mixed 
discrete/continuous path $\Gamma$ in the space of discrete and continuous variables $n_\nu$, $\xi$ and $\tau_\nu$. 
As a consequence of the closure relations \eqref{eq:Wronskrels} the kernel type quantity $G_{k,k'}$ is independent 
of the path $\Gamma$ and only depends on its end-points in the space of these variables. Thus, we have the following 
relations for its derivatives and differences: 
\[ \partial_\xi G_{k,k'}=\psi_{k'}\vf_k\  , \quad \partial_{\tau_\nu}G_{k,k'}=(1+u_{\tau_\nu})(T_p\psi_{k'})T_p^{-1}\vf_k\  , 
\quad \Delta_p G_{k,k'}=(T_p\psi_{k'})\vf_k\  .  \] 
 
Consider now the integral equations
\bse\label{eq:inteqs}
\begin{align}
& \vf_k=\vf_k^0+\ddint \vf_l\,G^0_{k,l'}\ , \\ 
& \psi_{k'}=\psi^0_{k'}+\ddint G_{l,k'} \psi^0_{l'}\ . 
\end{align}
\ese
where the integral is over a domain $D$ in any appropriate space of the parameters $k$ and $k'$ with an arbitrary measure 
${\rm d}\zeta(k,k')$ in that space\footnote{The integral over $k$ and $k'$ 
can be over a continuous domain $D$ or over a set of distinct values of these parameters.}, which is independent of the dynamical variables 
$n_\nu,\tau_\nu$ and $\xi$.  There is no need to specify this integral as long as some assumptions are valid, such as 
that operations like the lattice shifts $T_{p_\nu}$ and derivatives $\partial_\xi,\partial_{\tau_\nu}$ can be moved 
through the integral without problem. 

We can think of the integral equations \eqref{eq:inteqs} as an integral \textit{transform} from an initial solution $(u^0,\vf_k^0,\psi^0_{k'})$ 
of the system of equations (comprising the Lax pairs as well as the nonlinear equations for $u$) to a new 'dressed' solution 
$(u,\vf_k,\psi_{k'})$, where 
\begin{equation} \label{eq:utransf}
u=u^0-\ddint \vf_l\,\psi_{l'}^0\  . 
\end{equation} 
It was shown in \cite{FN16}, cf. also \cite{Nij85LMP} for the general matrix case, that this dressing 
transform has a group property as a consequence of the structure of the kernel function $G_{k,k'}$ 
as a path independent integrated object, cf. \eqref{eq:G}. Furthermore, the kernel itself 
obeys the dressing relation, namely  
\begin{equation} 
\label{eq:Gtransf}
G_{k,k'}=G_{k,k'}^0+\ddint G_{l,k'}\,G^0_{k,l'}\  .  
\end{equation} 

In the case of 'free reference' dressing, i.e. when $u^0=0$ and $\vf^0_k$ and $\psi^0_{k'}$ take the form \eqref{eq:freeref} we have the square eigenfunction expansions, cf. \cite{WCN87}, 
\bse\label{eq:sqeigen}
\begin{align}
& \partial u= \ddint \vf_l\psi_{l'}\,\partial\left(\ln(\vf^0_k\psi^0_{k'})\right)\  ,   \\ 
& Tu-u=\ddint (T\vf_l)\psi_{l'} \left(\frac{T\psi^0_{l'}}{\psi^0_{l'}}- \frac{\vf^0_l}{T\vf^0_l}\right)= 
\ddint \vf_l(T\psi_{l'}) \left(\frac{T\vf^0_{l}}{\vf^0_{l}}- \frac{\psi^0_{l'}}{T\psi^0_{l'}}\right)\ ,    
\end{align}
\ese
where $T$ is any shift operator which commutes with the double integral and $\partial$ is any first order derivative 
which commutes with the double integral. 

In particular, the following fundamental relation holds
\begin{equation}\label{eq:fundtau}
\ddint \frac{(l+l')\vf_l^0\psi^0_{l'}}{(a+l')(a'+l)}(T_{a'}^{-1}T_{-l}f)(T_{l'}^{-1}T_{-a}f) = 
f\,(T_{a'}^{-1}T_{-a}f)-(T_{a'}^{-1}f)(T_{-a}f)\  , 
\end{equation}
which is reminiscent of the bilinear identity that plays a central role in the work 
of the Kyoto school, and from which the hierarchy of Hirota bilinear 
equations can be derived, cf. \cite{MJD2000}.

\section{Discrete Lagrangian 3-form structure} 

We will now describe the Lagrangian structure for the KP system comprising both the 
fully discrete as well as the semi-discrete equations, in terms of Lagrangian difference 
or differential 3-forms. The most general structure of this form would be 
\begin{align}\label{eq:LKPgenmulti}
\mathsf{L}=& \sum_{i<j<k} \mathcal{L}_{p_ip_jp_k} \delta_{p_i}\wedge\delta_{p_j}\wedge\delta_{p_k}+ \mathcal{L}_{(\tau_i)(\tau_j)(\tau_k)}{\mathrm d}\tau_i\wedge{\mathrm d}\tau_j\wedge{\mathrm d}\tau_k \nonumber \\ 
& +\sum_i\sum_{j<k}\mathcal{L}_{(\tau_i)p_jp_k}{\mathrm d}\tau_i\wedge\delta_{p_j}\wedge\delta_{p_k}+\mathcal{L}_{p_i(\tau_j)(\tau_k)}\delta_{p_i}\wedge{\mathrm d}\tau_j\wedge{\mathrm d}\tau_k \nonumber \\  
& +\sum_{i<j} \mathcal{L}_{(\xi)p_ip_j}{\mathrm d}\xi\wedge\delta_{p_i}\wedge\delta_{p_j} 
+ \mathcal{L}_{(\xi)(\tau_i)(\tau_j)}{\mathrm d}\xi\wedge{\mathrm d}\tau_i\wedge{\rm d}\tau_j
\nonumber \\ 
&+\sum_{i,j} \mathcal{L}_{(\xi)(\tau_i)p_j}{\mathrm d}\xi\wedge{\rm d}\tau_i\wedge\delta_{p_j} \ ,
\end{align} 
which is to be integrated over a 3-dimensional hypersurface, containing smooth as well 
as semi-discrete and fully discrete coordinate patches. We will not attempt to describe 
such a weird hypersurface here, along the line as e.g. in \cite{Y-KLN,SleighVerm} for the much 
simpler situation of semi-discrete 1-curves and 2-surfaces, the main property needed 
here in the definition of an integral of the type 
\begin{equation}\label{eq:genaction} 
S[u(\xi,\boldsymbol{n},\boldsymbol{\tau});\mathcal{V}]=
%{\Huge{\sf{\bf{S}}}}_\Gamma \Omega_{k,k'} 
\mathrlap{\int}\sum_{\mathcal{V}} \mathsf{L}\  
\end{equation}
as a functional of both field variables $u(\xi,\boldsymbol{n},\boldsymbol{\tau})$ and 
semi-discrete hypersurfaces $\mathcal{V}$, is the existence of a generalised Stokes' 
theorem that applies to closed semi-discrete 
hypersurfaces. This can be established, but is technically (mostly as a matter of 
introducing the appropriate notations) quite involved, and we leave 
the formulation and proof of such a theorem to a future publication. The main 
point being that, as we want to find the critical point of the action \eqref{eq:genaction} 
for arbitrary hypersurfaces $\mathcal{V}$, it suffices to consider closed hypersurfaces 
and apply Stokes' theorem to the first order variation of $S$. Setting the latter to zero 
as stationarity condition, we arrive at the condition 
\be\label{eq:genEL} \delta{\mathrm d}\mathsf{L}=0 \ , \ee
which form the relevant generalised Euler-Lagrange equations where ${\mathrm d}$ denotes the relevant exterior difference/derivative of the Lagrangian 
3-form (in the language of  
a variational bi-complex with horizontal and vertical derivative operators ${\mathrm d}$ 
and $\delta$, the latter comprising discrete exterior `derivatives' as well as continuous ones). In order to develop the Lagrangian 3-form systematically, we break the process 
of deriving the EL equations and establishing the corresponding closure relation 
\be\label{eq:genclos}    {\mathrm d}\mathsf{L}\big|_{\rm EL}=0\ , \ee  
up in semi-discrete and fully discrete components, and treat them separately. 

In the case we restrict ourself to the fully discrete Lagrangian 
component we have a discrete Lagrangian 3-form
\be\label{eq:LKPmulti}
\mathsf{L}= \sum_{i<j<k} \mathcal{L}_{p_ip_jp_k} \delta_{p_i}\wedge\delta_{p_j}\wedge\delta_{p_k}\  ,
\ee
where the $\mathcal{L}_{p_ip_jp_k}$ each are the Lagrangian components for any three directions indicated by
the lattice parmeters $p_i$, $p_j$, $p_k$ (which could 
comprise the parameters $p$, $q$ and $r$ of the previous sections as 
particular choices) each associated with a discrete lattice variable $n_{p_i}$  playing the role of coordinates for the $i^{\rm th}$
direction in that multidimensional lattice. As before, in \eqref{eq:1form}
the $\delta_p$ denotes a \textit{discrete differential}\footnote{The notation is similar to the
one introduced in \cite{MansHyd}.}, i.e. the formal symbol 
indicating that in the restricted action functional 
\be\label{eq:action}
S[u(\boldsymbol{n});\nu]=\sum_{\nu} \mathsf{L}= \sum_{\nu_{ijk}\in\nu}
\mathcal{L}_{p_ip_jp_k} \delta_{p_i}\wedge\delta_{p_j}\wedge\delta_{p_k} 
\ee
the Lagrangian contributions from all elementary rhomboids $\nu_{ijk}=(\boldsymbol{n},\boldsymbol{n}+\boldsymbol{e}_i,
\boldsymbol{n}+\boldsymbol{e}_j,\boldsymbol{n}+\boldsymbol{e}_k)$ (with elementary displacement vectors $\boldsymbol{e}_i$ (see Figure 
\ref{nu_cube}) 
along the edges in the 3D lattice associated with the lattice parameter $p_i$) are simply
summed up according to their base point $\boldsymbol{n}$ and their orientation.

A Lagrangian for a 3-dimensional system can be defined on an elementary cube (or rhomboid) $\nu_{ijk}(\boldsymbol{n})$, where $\nu_{ijk}(\boldsymbol{n})$ is specified by the position $\boldsymbol{n}=(\boldsymbol{n},\boldsymbol{n}+\boldsymbol{e}_{i},\boldsymbol{n}+\boldsymbol{e}_{j},\boldsymbol{n}+\boldsymbol{e}_{k})$ of one of its vertices in the lattice and the lattice directions given by the base vectors $\boldsymbol{e}_{i},\boldsymbol{e}_{j},\boldsymbol{e}_{k}$, 
as in Figure \ref{nu_cube}. Thus, 
\begin{gather*}\label{3form} 
\mathcal{L}_{ijk}(\boldsymbol{n})= L(u(\boldsymbol{n}),u(\boldsymbol{n}+\boldsymbol{e}_{i}),u(\boldsymbol{n}+\boldsymbol{e}_{j}),u(\boldsymbol{n}+\boldsymbol{e}_{k}),
u(\boldsymbol{n}+\boldsymbol{e}_{i}+\boldsymbol{e}_{j}), \\ 
\qquad\qquad\qquad  u(\boldsymbol{n}+\boldsymbol{e}_{j}+\boldsymbol{e}_{k}),u(\boldsymbol{n}+\boldsymbol{e}_{i}+\boldsymbol{e}_{k}),u(\boldsymbol{n}+\boldsymbol{e}_{i}+\boldsymbol{e}_{j}+\boldsymbol{e}_{k});p_i,p_j,p_k)\ , 
\end{gather*}
as well as on lattice parameters $p_i,p_j,p_k$,  (where where for simplicity of notation we have replaced the 
indices $p_i,p_j,p_k$ simply by $i,j,k$ and where $\Delta_i=T_{p_i}-{\rm id}$ is the forward difference operator in the direction labeled by 
$p_i$). 

According to the derivation proposed in \cite{LobbNij2018}, the set of multiform EL equations is
obtained by considering the smallest closed 3D hypersurface, which is the boundary of a 4D hypercube (hyper-rhomboid) $\mathcal{R}$,
the action of which is given by
\be\label{eq:boxL}
S[u(\boldsymbol{n});\partial\mathcal{R}]=:(\square\mathcal{L})_{ijkl}=
\Delta_{i}\mathcal{L}_{jkl}- \Delta_j\mathcal{L}_{kli}+\Delta_k\mathcal{L}_{lij}\ 
-\Delta_l\mathcal{L}_{ijk}\ .
\ee
Embedding the system in 4 dimension, the smallest closed 3-dimensional 
hypersurface is the boundary of the tessaract, consisting of 8 3-dimensional faces, each of which is a rhomboid. 
Because of the symmetry, we need only take derivatives with respect to $u,T_{i}u$, $T_{i}T_{j}u$, $T_{i}T_{j}T_{k}u$ and $T_{i}T_{j}T_{k}T_{l}u$, and the other equations will follow by cyclic permutation of the lattice directions.

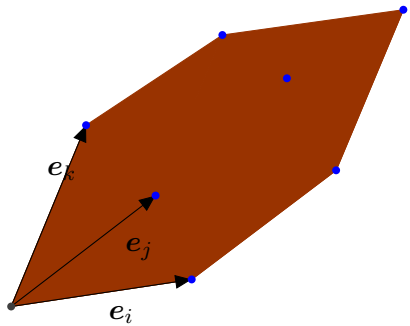
\begin{figure}[h]
\begin{center}
%%%%%%%%%%%%%%%%%%
\definecolor{zzttqq}{rgb}{0.6,0.2,0}
\definecolor{qqqqff}{rgb}{0,0,1}
\definecolor{uququq}{rgb}{0.25,0.25,0.25}
%\begin{tikzpicture}[line cap=round,line join=round,>=triangle 45,x=1.0cm,y=1.0cm]
\begin{tikzpicture}[line cap=round,line join=round,>=triangle 45,x=.6cm,y=.6cm]
\fill[color=zzttqq,fill=zzttqq,fill opacity=0.1] (0,0) -- (3.2,2.46) -- (7.2,3.02) -- (4,0.6) -- cycle;
\fill[color=zzttqq,fill=zzttqq,fill opacity=0.1] (0,0) -- (1.66,4.02) -- (4.68,6.02) -- (3.2,2.46) -- cycle;
\fill[color=zzttqq,fill=zzttqq,fill opacity=0.1] (1.66,4.02) -- (6.11,5.06) -- (4,0.6) -- (0,0) -- cycle;
\fill[color=zzttqq,fill=zzttqq,fill opacity=0.1] (7.2,3.02) -- (8.69,6.58) -- (4.68,6.02) -- (3.2,2.46) -- cycle;
\fill[color=zzttqq,fill=zzttqq,fill opacity=0.1] (8.69,6.58) -- (6.11,5.06) -- (4,0.6) -- (7.2,3.02) -- cycle;
\draw [color=zzttqq] (0,0)-- (3.2,2.46);
\draw [color=zzttqq] (3.2,2.46)-- (7.2,3.02);
\draw [color=zzttqq] (7.2,3.02)-- (4,0.6);
\draw [color=zzttqq] (4,0.6)-- (0,0);
\draw [color=zzttqq] (0,0)-- (1.66,4.02);
\draw [color=zzttqq] (1.66,4.02)-- (4.68,6.02);
\draw [color=zzttqq] (4.68,6.02)-- (3.2,2.46);
\draw [color=zzttqq] (3.2,2.46)-- (0,0);
\draw [color=zzttqq] (1.66,4.02)-- (6.11,5.06);
\draw [color=zzttqq] (6.11,5.06)-- (4,0.6);
\draw [color=zzttqq] (4,0.6)-- (0,0);
\draw [color=zzttqq] (0,0)-- (1.66,4.02);
\draw [color=zzttqq] (7.2,3.02)-- (8.69,6.58);
\draw [color=zzttqq] (8.69,6.58)-- (4.68,6.02);
\draw [color=zzttqq] (4.68,6.02)-- (3.2,2.46);
\draw [color=zzttqq] (3.2,2.46)-- (7.2,3.02);
\draw [color=zzttqq] (8.69,6.58)-- (6.11,5.06);
\draw [color=zzttqq] (6.11,5.06)-- (4,0.6);
\draw [color=zzttqq] (4,0.6)-- (7.2,3.02);
\draw [color=zzttqq] (7.2,3.02)-- (8.69,6.58);
\draw [->] (0,0) -- (4,0.6);
\draw [->] (0,0) -- (3.2,2.46);
\draw [->] (0,0) -- (1.66,4.02);
\draw (2.31,1.74) node[anchor=north west] {$ \boldsymbol{e}_j $};
%\draw (2.80,1.74) node[anchor=north west] {$ \boldsymbol{e}_i $};

\draw (1.94,.24) node[anchor=north west] {$ \boldsymbol{e}_i $};
\draw (0.57,3.43) node[anchor=north west] {$ \boldsymbol{e}_k $};
\begin{scriptsize}
\fill [color=uququq] (0,0) circle (1.5pt);
\draw[color=uququq] (0.24,0.42) node {$ $};
\fill [color=qqqqff] (3.2,2.46) circle (1.5pt);
%\draw[color=qqqqff] (3.44,2.87) node {$ $};
\fill [color=qqqqff] (7.2,3.02) circle (1.5pt);
%\draw[color=qqqqff] (7.47,3.43) node {$C$};
\fill [color=qqqqff] (4,0.6) circle (1.5pt);
%\draw[color=qqqqff] (4.26,1.01) node {$D$};
\fill [color=qqqqff] (1.66,4.02) circle (1.5pt);
%\draw[color=qqqqff] (1.9,4.44) node {$E$};
\fill [color=qqqqff] (4.68,6.02) circle (1.5pt);
%\draw[color=qqqqff] (4.89,6.42) node {$F$};
\fill [color=qqqqff] (6.11,5.06) circle (1.5pt);
%\draw[color=qqqqff] (6.37,5.47) node {$G$};
\fill [color=qqqqff] (8.69,6.58) circle (1.5pt);
%\draw[color=qqqqff] (8.85,6.98) node {$I$};
\end{scriptsize}
\end{tikzpicture}
%%%%%%%%%%%%%%%%%%
\caption{Elementary oriented rhomboid.}
\label{nu_cube}
\end{center}
\end{figure}

The variations with regard to the vertices of the closed 
hyper-rhomboid lead to the following set of equations, cf. \cite{LobbNij2018}, 
\begin{subequations}\label{3deqn} 
\begin{gather}
 0  =  \frac{\partial}{\partial u}\biggl(-\mathcal{L}_{ijk}+\mathcal{L}_{jkl}-\mathcal{L}_{kli}+\mathcal{L}_{lij}\biggr),\label{3deqna}\\
 0  =  \frac{\partial}{\partial T_{i}u}\biggl(-\mathcal{L}_{ijk}-T_{i}\mathcal{L}_{jkl}-\mathcal{L}_{kli}+\mathcal{L}_{lij}\biggr),\label{3deqnb}\\
 0  =  \frac{\partial}{\partial T_{i}T_{j}u}\biggl(-\mathcal{L}_{ijk}-T_{i}\mathcal{L}_{jkl}+T_{j}\mathcal{L}_{kli}+\mathcal{L}_{lij}\biggr),\label{3deqnc}\\
 0  =  \frac{\partial}{\partial T_{i}T_{j}T_{k}u}\biggl(-\mathcal{L}_{ijk}-T_{i}\mathcal{L}_{jkl}+T_{j}\mathcal{L}_{kli}-T_{k}\mathcal{L}_{lij}\biggr),\label{3deqnd}\\
 0  =  \frac{\partial}{\partial T_{i}T_{j}T_{k}T_{l}u}\biggl(T_{l}\mathcal{L}_{ijk}-T_{i}\mathcal{L}_{jkl}+T_{j}\mathcal{L}_{kli}-T_{k}\mathcal{L}_{lij}\biggr),\label{3deqne}
\end{gather}
\end{subequations}
along with the equivalent shifted versions
\begin{subequations}\label{3deqnshift}
\begin{gather}
 0  =  \frac{\partial}{\partial u}\biggl(-\mathcal{L}_{ijk}+\mathcal{L}_{jkl}-\mathcal{L}_{kli}+\mathcal{L}_{lij}\biggr),\\
 0  =  \frac{\partial}{\partial u}\biggl(-T_{i}^{-1}\mathcal{L}_{ijk}-\mathcal{L}_{jkl}-T_{i}^{-1}\mathcal{L}_{kli}+T_{i}^{-1}\mathcal{L}_{lij}\biggr),\\
 0  =  \frac{\partial}{\partial u}\biggl(-T_{i}^{-1}T_{j}^{-1}\mathcal{L}_{ijk}-T_{j}^{-1}\mathcal{L}_{jkl}+T_{i}^{-1}\mathcal{L}_{kli}+T_{i}^{-1}T_{j}^{-1}\mathcal{L}_{lij}\biggr),\\
 0  =  \frac{\partial}{\partial u}\biggl(-T_{i}^{-1}T_{j}^{-1}T_{k}^{-1}\mathcal{L}_{ijk}-T_{j}^{-1}T_{k}^{-1}\mathcal{L}_{jkl}+T_{i}^{-1}T_{k}^{-1}\mathcal{L}_{kli}-T_{i}^{-1}T_{j}^{-1}\mathcal{L}_{lij}\biggr),\\
 0  =  \frac{\partial}{\partial u}\biggl(T_{i}^{-1}T_{j}^{-1}T_{k}^{-1}\mathcal{L}_{ijk}-T_{j}^{-1}T_{k}^{-1}T_{l}^{-1}\mathcal{L}_{jkl}+T_{i}^{-1}T_{k}^{-1}T_{l}^{-1}\mathcal{L}_{kli}-T_{i}^{-1}T_{j}^{-1}T_{l}^{-1}\mathcal{L}_{lij}\biggr).
\end{gather}
\end{subequations}
The system of `corner equations' comprising \eqref{3deqn}, \eqref{3deqnshift} supplemented with the 3-form closure relation 
\begin{equation}\label{3dclosure} 
 \Delta_{l}\mathcal{L}_{ijk}-\Delta_{i}\mathcal{L}_{jkl}+\Delta_{j}\mathcal{L}_{kli}-\Delta_{k}\mathcal{L}_{lij}=0 
\end{equation}  
forms the fundamental system of EL equations of a 3D integrable fully discrete system, when we consider it as a system which not only represents the 
relevant equations of motion, but also as a system of equations for 
the (possibly unknown) Lagrangian components themselves. However, in 
applying the multiform scheme to the KP system, we will make a short-cut 
by regarding the double-zero structure of ${\mathrm d}\mathsf{L}$, following \cite{SNC22,NijZhang23,RichVerm}, cf. also \cite{Sleigh-thesis},  
 which provides a short-cut to the relevant EL system in the sense of 
the variational bi-complex, which implies the closure relation 
\eqref{3dclosure}.

\section{Lagrangian 3-forms for (semi-)discrete KP systems} 

A main aim of this paper is to describe the Lagrangian 3-form structure 
for the fully discrete KP system \eqref{eq:potDDDKP}, whereas the only instance so far of a Lagrangian structure for a fully discrete 3-dimensional integrable equation was given in \cite{LNQ}, 
namely for the Hirota bilinear KP equation, cf. also 
\cite{BolPetSur2015} for a geometric interpretation of that result. 
In order to arrive at that end, we will take an indirect route, namely 
via a Lagrange structure for the semi-discrete KP stystem \eqref{eq:DddKP}. 
Furthermore, we will establish the double-zero structure of the 
fully discrete Lagrangians 3-form.

\subsection{Lagrangian structure of the semi-discrete KP} 

A Lagrangian for eq. \eqref{eq:DddKP} is given by 
\begin{equation}\label{eq:DddLagr}
\mathcal{L}_{(\xi)pq}=v_\xi \ln(p-q+\wh{u}-\wt{u})+v(\wh{\wt{u}}+u-\wh{u}-\wt{u})\  , 
\end{equation} 
where $u$ and $v$ are independently variable fields. The conventional 
Euler-Lagrange (EL) equations (as a semi-discrete field theory) are given by 
\bse\begin{align}
& \frac{\delta \mathcal{L}_{(\xi)pq}}{\delta v}=-\partial_\xi\ln(p-q+\wh{u}-\wt{u})+\wh{\wt{u}}+u-\wh{u}-\wt{u}=0\ , \label{eq:DddELa}\\ 
& \frac{\delta \mathcal{L}_{(\xi)pq}}{\delta u}= \frac{\uh{v}_\xi}{p-q+u-\wt{\uh{u}}}-\frac{\ut{v}_\xi}{p-q+\wh{\ut{u}}-u}+\uh{\ut{v}}+v 
-\uh{v}-\ut{v}=0\ , \label{eq:DddELb}
\end{align}\ese
where \eqref{eq:DddELa} is \eqref{eq:DddKP}, and \eqref{eq:DddELb} a direct consequence provided we have the solution $v=\wh{u}-\wt{u}+{\rm constant}$. Note that the latter is only valid 'on-shell' (i.e. on solutions of the EL equations), as 
posing this off-shell (i.e., as a reduction on the variational system) would trivialize the Lagrangian \eqref{eq:DddLagr}. 
Lagrangians for the DD$\Delta$ equations \eqref{eq:DDdKP} are given by 
\bse\label{eq:DDdLagr}\begin{align} 
&\mathcal{L}_{(\xi)(\tau)p} = V_\xi\ln(1+u_\tau)+V(\wt{u}+\ut{u}-2u)\  , 
\label{eq:DDdLagra} \\ 
&\mathcal{L}_{(\xi)(\sigma)q} = W_\xi\ln(1+u_\sigma)+W(\wh{u}+\uh{u}-2u)\  , \label{eq:DDdLagrb}
\end{align}\ese 
the (conventional) EL equations of which yield  
\bse\begin{align}
& \frac{\delta \mathcal{L}_{(\xi)(\tau)p}}{\delta V}=-\partial_\xi\ln(1+u_{\tau}) 
+\wt{u}+\ut{u}-2u=0\ , \\ 
& \frac{\delta \mathcal{L}_{(\xi)(\tau)p}}{\delta u}= -\partial_\tau\left( 
\frac{V_\xi}{1+u_\tau}\right) +\wt{V}+\ut{V}-2V=0\ .
\end{align}\ese
(and similar relations arising from \eqref{eq:DDdLagrb}). Thus, varying 
w.r.t. $V$ (resp. $W$) yield the equations \eqref{eq:DddKP}, while 
the variations with 
respect to $u$ yield EL equations that are satisfied by $V=1+u_\tau$ and $W=1+u_\sigma$ respectively on-shell. 

\paragraph{\bf Remark:} 
We note that there is an (on-shell) closure relation 
\begin{align}\label{eq:VWclos} 
\left(\partial_\sigma \mathcal{L}_{(\xi)(\tau)p}
-\partial_\tau\mathcal{L}_{(\xi)(\sigma)q}\right)\Big{|}_{\rm EL}=\partial_\xi \left[u_{\tau\sigma}\ln\left(\frac{1+u_\tau}{1+u_\sigma} \right) \right]\  , 
\end{align}
where the expression within the straight brackets on the r.h.s. can be considered as a `null Lagrangian' (one whose conventional 
EL equations are trivial). 
The closure relation \eqref{eq:VWclos} absurdly suggests that there might be a Lagrangian multi-form 
\be\label{eq:sigtau3form} 
\mathsf{L}= \mathcal{L}_{(\xi)(\tau)p}{\mathrm d}\xi\wedge{\mathrm d}\tau\wedge\delta_p+ 
\mathcal{L}_{(\xi)(\sigma)q}{\mathrm d}\xi\wedge{\mathrm d}\sigma\wedge\delta_q+\mathcal{L}_{(\tau)(\sigma)}{\mathrm d}\tau\wedge{\mathrm d}\sigma \ , 
\ee 
which would be a mixed 3-form and 2-form. Whether such an object would  
make sense should perhaps not be rejected off-hand (the latter term 
representing perhaps some boundary term) but for now we postpone the 
investigation of the corresponding multiform structure. Instead, below we 
will develop the multiform structure associated with \eqref{eq:DddLagr}. 

\subsection{Lagrangians for the Lax representation} 

Before discussing the multiform structure for \eqref{eq:DddKP}, we make a 
side-step by showing that there are Lagrangians also for the corresponding 
Lax represenations of \eqref{eq:DddKP} and \eqref{eq:DDdKP}, namely by 
using the square eigenfunction representations \eqref{eq:sqeigen}. It is 
here where the multiform structure becomes crucial, as a Lax representation 
requires the posing of an overdetermined pair of equations, rather than 
a single equation (the Lax equation). In this spirit a variational 
description of Lax \textit{pairs} was first given in \cite{SNC20}.  

The individual components of the Lax representation, namely \eqref{eq:vfLaxa} and \eqref{eq:DdvfLaxa} and their respective adjoints 
\eqref{eq:psiLaxa} and \eqref{eq:DdpsiLaxa}, can be obtained from the 
following Lagrangians respectively,   
\bse\label{eq:LaxLagrs}
\begin{align}
& \mathsf{L}_p = \tfrac{1}{2}(\wt{u}-u)^2 +p(u-\wt{u})
+\ddint (l+l') \left[ \wt{\psi}_{l'}\partial_\xi\vf_l-\psi_{l'}\vf_l+\wt{\psi}_{l'}(p+u-\wt{u})\vf_l\right] \ , \\ 
&\mathsf{L}_{(\tau)}= u\partial_\tau\wt{u}+\ddint (l+l')(2p+l-l') 
\left[ (1+u_\tau)\wt{\psi}_{l'}\,\ut{\vf}_l+\wt{\psi}_{l'}\partial_\tau\vf_l \right] \ , 
\end{align}
\ese 
and similar Lagrangians associated with the other lattice direction and with the variable $\sigma$ and parameter $q$ instead of $\tau$ and $p$.  
The verification that these Lagrangians yield the correct Lax pair equations utilises an application of the 
formulae \eqref{eq:sqeigen} in the case that $T=T_p$ and $\partial=\partial_\tau$. In fact, for those choices of 
shift and differential operator, we get by using also \eqref{eq:freeref} the square eigenfunction expansion 
\[ \wt{u}-u=\ddint (l+l') \wt{\psi}_{l'}\vf_l\ , \quad \wt{u}-\ut{u}= \ddint (l+l')(2p+l-l') \wt{\psi}_{l'}\ut{\vf}_l\  . \]  
and 
\[ u_\xi=\ddint (l+l')\psi_{l'}\vf_l\  , \quad   u_\tau=\ddint \frac{l+l'}{(p+l)(p-l')} \psi_{l'}\vf_l\  .   \]
Thus, variations of \eqref{eq:LaxLagrs} yields 
\bse\label{eq:LaxELs1} 
\begin{align}
& \frac{\delta \mathsf{L}_p}{\delta\psi_{k'}}= \int_{C_{k'}} {\rm d}\lambda_{k'}(l) 
\left[-\vf_l+\partial_\xi\ut{\vf}_l+(p+\ut{u}-u)\ut{\vf}_l  \right]=0 \  , \\ 
& \frac{\delta \mathsf{L}_p}{\delta\vf_{k}}= \int_{C'_{k}} {\rm d}\lambda'_{k}(l') 
\left[-\psi_{l'}-\partial_\xi\wt{\psi}_{l'}+(p+u-\wt{u})\wt{\psi}_{l'}  \right]=0 \  , \\ 
& \frac{\delta \mathsf{L}_p}{\delta u}= 2u-\wt{u}-\ut{u} + 
\ddint (l+l') \left(\wt{\psi}_{l'}\vf_l-\psi_{l'}\ut{\vf}_l\right)=0\ ,    
\end{align}
\ese 
where the latter holds true by virtue of the square eigenfunction expansion, whilst the 
former two are weak forms of the Lax relations. Note that we have assumed that the functional 
derivative with regard to the $k$- and $k'$-dependent quantity within the double integral reduces the double  
integral to a single integral over (boundary) curves $C_{k'}$ and $C'_k$ in the space of 
$k$ respectively $k'$ variables\footnote{This would obviously require an analytic 
justification, for specific integration regions $D$ and measure ${\rm d}
\zeta$, but that is beyond the scope of this paper. Here we will just 
assume that this can be done for specific choices of the integrations.}. 

Variations of $\mathsf{L}_{(\tau)}$ yield the following EL equations 
\bse\label{eq:LaxELs2} 
\begin{align}
& \frac{\delta \mathsf{L}_{(\tau)}}{\delta\psi_{k'}}= \int_{C_{k'}} {\rm d}\lambda_{k'}(l)\, (2p+l-l') 
\left[(1+\ut{u}_\tau)\ut{\ut{\vf}}_l+\partial_\tau\ut{\vf}_l  \right]=0 \  , \\ 
& \frac{\delta \mathsf{L}_{(\tau)}}{\delta\vf_{k}}= \int_{C'_{k}} {\rm d}\lambda'_{k}(l')\, (2p+l-l') 
\left[(1+\wt{u}_\tau)\wt{\wt{\psi}}_{l'}-\partial_\tau\wt{\psi}_{l'} \right]=0 \  , \\ 
& \frac{\delta \mathsf{L}_{(\tau)}}{\delta u}= \partial_\tau(\wt{u}-\ut{u}) -\partial_\tau  
\ddint (l+l') (2p+l-l')\wt{\psi}_{l'}\ut{\vf}_l=0\ ,    
\end{align}
\ese 
which are in some sense weaker forms of the Lax pair \eqref{eq:DdpsiLax} and where the latter holds true by virtue of the square eigenfunction expansion. The main question is how to fit the Lagrangians $\mathsf{L}_p$ 
and $\mathsf{L}_q$, or alternatively the Lagrangians $\mathsf{L}_{\tau}$ 
and $\mathsf{L}_{\sigma}$, into a multiform so that the multiform EL 
equation produce the respective pairs. This is still an open problem, but 
the following observation may provide a lead. 

\paragraph{\bf Remark1:} 
The eigenfunction Lagrangians \eqref{eq:LaxLagrs} obey the following 
closure relation 
\[ \left(\Delta_p \mathsf{L}_q-\Delta_q\mathsf{L}_p\right)\,
\Big{|}_{\rm EL}=\partial_\xi(p-q+\wh{u}-\wt{u})\  . \]  
This suggests a Lagrangian 2-form
\[ \mathsf{\bf L}=\mathsf{L}_p \delta_p\wedge {\rm d}\xi +
\mathsf{L}_q\delta_q\wedge{\rm d}\xi+\mathsf{L}_{pq}\delta_p\wedge \delta_q\  ,     \] 
where $\mathsf{L}_{pq}=p-q+\wh{u}-\wt{u}$ is a \textit{null Lagrangian} which has trivial EL equations. However, this closure relation does not seem yet either  
sufficient or appropriate to provide the 
full variational structure for the Lax representation, and we intend to come back to 
this matter in a future publication. 

\paragraph{\bf Remark2:}  
Although the Lax equations come out in what seems 
to be a weaker (i.e., integrated) form, they are really are not in any sense weaker as for instance in the case of \eqref{eq:LaxELs1}  
one could introduce the functions 
\[ \Phi_{k'}=\int_{C_{k'}} {\rm d}\lambda_{k'}(l)\vf_{l}\ , \quad 
\Psi_k= \int_{C'_{k}} {\rm d}\lambda'_{k}(l')\psi_{l'}\  ,  \] 
in terms of which we recover the original Lax pairs both the adjoint and the direct forms. In the case of \eqref{eq:LaxELs2} it is more subtle, 
but essentially the variational form of the Lax is as strong as the 
original form. 
 
The Lax equations for the semi-discrete KP system are thus retrievable from well-chosen Lagrangians, but 
what we need is a mechanism to obtain these equations simultaneously as Lax pairs from variations of 
a single structure. For this we need the notion of Lagrangian multiforms, which we will address now in the case of Lagrangains of the form 
\eqref{eq:DddLagr}.

\subsection{From semi-discrete to fully discrete KP Lagrangian} 

We will now first establish the Lagrangian multiform structure for the semi-discrete KP system. Surprisingly this 
will provide us with a Lagrangian structure for the fully discrete KP equation \eqref{eq:dddKP}. 
%\footnote{In \cite{LNQ} the Lagrangian multiform structure was established for the 
%Hirota form of the KP equation, the lattice geometry for which was subsequently analysed in \cite{BollPetrSur}.}. 
Considering three copies of the Lagrangian \eqref{eq:DddLagr} in three lattice dimensions (in addition to the continuous one associated 
with $\xi$), and observing 
that \textit{on-shell} the variable $v$ can be identified with $v\dot{=}p-q+\wh{u}-\wt{u}$, we note that the 
latter quantity is actually a `2-form valued object'. Thus, for each pair of lattice directions we need to 
introduce a separate variable $v$, and consequently in three dimensions the Lagrangian components can be set as: 
\bse\label{eq:DddKPLagrcomps} 
\begin{align}
& \mathcal{L}_{(\xi)pq}=v_\xi \ln(p-q+\wh{u}-\wt{u})+v(\wh{\wt{u}}+u-\wh{u}-\wt{u})\  , \\ 
& \mathcal{L}_{(\xi)qr}=w_\xi \ln(q-r+\wb{u}-\wh{u})+w(\wh{\wb{u}}+u-\wh{u}-\wb{u})\  , \\ 
& \mathcal{L}_{(\xi)rp}=z_\xi \ln(r-p+\wt{u}-\wb{u})+z(\wt{\wb{u}}+u-\wt{u}-\wb{u})\  , 
\end{align}\ese 
where on-shell the additional 2-form variables $w$ and $z$ can be identified with 
\[ w=q-r+\wb{u}-\wh{u}\ , \quad z=r-p+\wt{u}-\wb{u}\  . \] 
These are constrained by the conditions 
\begin{equation}\label{eq:vwzconstrs} 
v+w+z=0\ , \quad \wb{v}+\wt{w}+\wh{z}=0\  .  
\end{equation} 
Furthermore, we note that the potential lattice KP equation \eqref{eq:dddKP} can be written as 
\begin{equation}\label{eq:vwzKP}
\frac{\wb{v}}{v}=\frac{\wt{w}}{w}=\frac{\wh{z}}{z}\  .  
\end{equation} 
The system comprising \eqref{eq:vwzconstrs} and \eqref{eq:vwzKP} was 
considered in the context of the non-Abelian case by \cite{Nimmo}, while 
in \cite{FN16},  we showed how by elimination 
this system 
gives rise to a scalar \textit{non-potential} lattice KP equation for each 
of the fields $v$, $w$ or $z$. For instance, in terms of $v$ 
the resulting non-potential lattice KP is the 10-point equation 
\begin{align}\label{eq:nonpotKP} 
& \wh{\wt{v}}\wh{\wh{v}}\left[\left(\wh{\wt{\wb{v}}}-\wh{\wh{\wt{v}}}\right)\wh{v} \wt{\wb{v}}+\wh{\wh{\wt{v}}}\wt{v}\wh{\wb{v}} \right]
+\wh{\wt{v}}\wh{\wh{\wb{v}}} \left[\wh{\wb{v}}\wt{v}\left(\wh{\wb{v}}-\wh{\wh{v}}\right)+\wh{\wh{v}}\wh{v}\wt{\wb{v}} \right]   \\ 
& = \wh{\wb{v}}\wt{\wb{v}}\left[ \wh{\wt{v}} \wh{\wh{\wb{v}}}-\wh{\wh{v}} \wh{\wt{\wb{v}}}\right](\wb{v}-\wt{v})
+\wh{\wb{v}}\wh{\wh{v}}\wh{\wt{\wb{v}}}\wt{\wb{v}}\wh{v}+\wh{\wb{v}}\wh{\wt{v}}^2\wt{v} \wh{\wh{\wb{v}}}\ . 
\end{align}
The key indication to 
establish a multiform structure has been to establish a \textit{closure relation} between 
the Lagrangian components of a MDC integrable system, \cite{LobbNij2009}. In the present case this is given by 
the following theorem. 

\begin{theorem}
Between the Lagrangian components \eqref{eq:DddKPLagrcomps} the following 3-form closure 
relation holds true on-shell (i.e. for solutions of the equations \eqref{eq:vwzconstrs} 
and \eqref{eq:vwzKP}): 
\bse\label{eq:KPclos}
\begin{equation}\label{eq:closure} 
\Delta_r \mathcal{L}_{(\xi)pq}+\Delta_p \mathcal{L}_{(\xi)qr}+\Delta_q \mathcal{L}_{(\xi)rp}=
\partial_\xi \mathcal{L}_{pqr}\  , 
\end{equation}  
where 
\begin{equation}\label{eq:dddKPLagr}
\mathcal{L}_{pqr}=(\wb{v}-v)\ln v+(\wt{w}-w)\ln w+(\wh{z}-z)\ln z\  . 
\end{equation} 
\ese 
Furthermore, the (conventional) Euler-Lagrange equations for the constrained Lagrangian component 
\begin{equation} 
\mathcal{L}^{\rm c}_{pqr}=\mathcal{L}_{pqr}+\lambda(v+w+z)+\mu(\wb{v}+\wt{w}+\wh{z})\  , 
\end{equation} 
where $\lambda$ and $\mu$ are Lagrange multipliers, are satisfied on solutions of the lattice  
potential KP equation \eqref{eq:vwzKP}. 
\end{theorem} 

\begin{proof}
First we note that on solutions of the EL equations, i.e. on solutions of \eqref{eq:DddKP} 
we have 
\[ \left.\mathcal{L}_{(\xi)pq}\right|_{\rm EL}=\partial_\xi\left(v\ln(p-q+\wh{u}-\wt{u}) \right)\ . \] 
Thus by writing out the left-hand side of \eqref{eq:closure} on-shell we get 
a total derivative $\partial_\xi$ of a sum of terms that can be rewritten in the form 
of \eqref{eq:dddKPLagr}.   
Furthermore, the EL equations for $\mathcal{L}_{pqr}$, varying with respect to $v,w$ and $z$ 
independently, are
\bse\label{eq:vwzEL}
\begin{align}
& \frac{\delta \mathcal{L}^c}{\delta v}=\ln\left(\ub{v}/v\right)+\frac{\wb{v}}{v}+\lambda-1+\ub{\mu}=0\ , \\ 
& \frac{\delta \mathcal{L}^c}{\delta w}=\ln\left(\ut{w}/w\right)+\frac{\wt{w}}{w}+\lambda-1+\ut{\mu}=0\ , \\ 
& \frac{\delta \mathcal{L}^c}{\delta z}=\ln\left(\uh{z}/z\right)+\frac{\wh{z}}{z}+\lambda-1+\uh{\mu}=0\ , 
\end{align}\ese
together with the constraints \eqref{eq:vwzconstrs}. A natural solution to these EL equations 
is 
\[ \frac{\wb{v}}{v}=\frac{\wt{w}}{w}=\frac{\wh{z}}{z}=e^{-\mu}=1-\lambda\ ,  \] 
which supplemented with the constraints \eqref{eq:vwzconstrs} yield copies of the non-potential lattice KP equation.  
\end{proof} 

Theorem 6.1 implies that the semi-discrete Lagrangian 3-form 
\begin{equation}\label{eq:Lagr3form}
\mathsf{L}=\mathcal{L}_{pqr}\,\delta_p\wedge\delta_q\wedge\delta_r + \mathcal{L}_{(\xi)pq}\,{\mathrm d}\xi\wedge\delta_p\wedge\delta_q + \mathcal{L}_{(\xi)qr}\,{\mathrm d}\xi\wedge\delta_q\wedge\delta_r 
+ \mathcal{L}_{(\xi)rp}\,{\mathrm d}\xi\wedge\delta_r\wedge\delta_p  
\end{equation} 
is closed on solutions of the system of semi-discrete and fully discrete KP equations. 

\paragraph{Remark:} 
A theory of 
semi-discrete Lagrangian 2-forms was recently proposed in \cite{SleighVerm}, referring to earlier work 
\cite{MansHyd} on a general theory of \textit{difference forms}, cf. also \cite{Sleigh-thesis}. In \cite{SleighVerm} a description was given of 
semi-discrete 2-manifold, while in \cite{Y-KLN} (in the special case of 
Calogero-Moser type systems) we gave a description of semi-discrete 
curves. Although the general case of semi-discrete Lagrangian multiforms was outlined in the Appendix of \cite{SleighVerm}, no examples were given 
of a system with more than one discrete variable. Thus, the semi-
discrete KP case with the semi-discrete 3-form \eqref{eq:Lagr3form}, 
seems to be the first concrete example of a higher-dimensional semi-discrete multiform with more than one discrete 
direction involved in corresponding the semi-discrete multi-time manifold. 
\vspace{.2cm} 

However, this is not the end of the story. In fact, the treatment above also gives rise to 
a fully discrete 3-form structure comprising the Lagrangian components of the form 
\eqref{eq:dddKPLagr} embedded in a multidimensional lattice of at least four lattice 
directions. This is given by the following:

\begin{theorem} 
The fully discrete Lagrangian 3-form 
\begin{equation}\label{eq:Lagr3dddform}
\mathsf{L}=\mathcal{L}_{pqr}\,\delta_p\wedge\delta_q\wedge\delta_r + \mathcal{L}_{qrs}\,
\delta_q\wedge\delta_r\wedge\delta_s + \mathcal{L}_{rsp}\,\delta_r\wedge\delta_s\wedge 
\delta_p + \mathcal{L}_{spq}\,\delta_s\wedge\delta_p\wedge\delta_q 
\end{equation} 
where the components are given by \eqref{eq:dddKPLagr} is closed on solutions of the 
fully discrete lattice KP equation \eqref{eq:dddKP}, i.e. it obeys the following 
3-form closure relation on solutions: 
\begin{equation}\label{eq:dddclosure} 
(\square\mathcal{L})_{pqrs}:=\Delta_s \mathcal{L}_{pqr}-\Delta_p \mathcal{L}_{qrs}+\Delta_q \mathcal{L}_{rsp}- 
\Delta_r \mathcal{L}_{spq}=0\  . 
\end{equation} 
\end{theorem} 
\begin{proof}
The proof is by direct computation. It takes into account that each component of the 3-form 
\eqref{eq:Lagr3dddform} is defined on a rhomboid in higher dimension, each face of which 
is associated with a variable $v_{p_ip_j}$ (where $v_{pq}=v$, $v_{qr}=w$ and 
$v_{rp}=z$ in the ad-hoc notation used earlier), which on solutions of the EL equations can be identified with 
$p_i-p_j+T_{p_j}u-T_{p_i}u$, and hence subject to the relations 
\[ v_{p_ip_j}=-v_{p_jp_i}\  , \quad v_{p_ip_j}+v_{p_jp_k}+v_{p_kp_i}=0\ , \quad 
\frac{T_{p_k}v_{p_ip_j}}{v_{p_ip_j}}= \frac{T_{p_i}v_{p_kp_j}}{v_{p_kp_j}}\  , 
\]
the latter relation being the lattice potential KP equation on each rhomboid.  
Expanding the left-hand side of \eqref{eq:dddclosure} there are various types of terms and their 
permutations of indices. The terms of the type $T_{p_i}\left[\left(T_{p_j}v_{p_kp_l}\right)\ln(v_{p_kp_l})\right]$ 
all cancel against each other using the above identification for the $v_{p_kp_l}$ and the 
lattice potential KP equation. The terms of the type $v_{p_kp_l}\ln(v_{p_kp_l})$ cancel out by the skew symmetry 
among the cyclic permutations of the indices. The mixed shifted terms of the types 
$T_{p_i}\left[v_{p_kp_l}\ln(v_{p_kp_l})\right]$ and $\left(T_{p_j}v_{p_kp_l}\right)\ln(v_{p_kp_l})$ 
conspire together to cancel, using the lattice potential KP and the form of the quantities $v_{p_ip_j}$, where there 
are two types of cyclic permutations, even and odd, that provide separate cancellations of the contributions 
arising in the closure.    
\end{proof} 

There is an even stronger result that pertains to the double-zero 
phenomenon of the exterior derivative of the Lagrange multiform, following 
the ideas in \cite{SNC22,RichVerm,NijZhang23}, which is stated as follows. 

\begin{theorem} 
The discrete exterior derivative of the Lagrangian 3-form \eqref{eq:Lagr3dddform} possesses the double-zero property, meaning that \textit{off-shell} it can be written as a sum of 
factors, each of which vanish on the solutions of the EL equations. 
\end{theorem}

\begin{proof}
The proof is by direct computation. Writing the Lagrangian components as 
\begin{align}\label{eq:dddKPLagrcomp}
\mathcal{L}_{pqr}=(T_rv_{pq}-v_{pq})\ln( v_{pq})+(T_pv_{qr}-v_{qr})\ln(v_{qr}) 
+(T_qv_{rp}-v_{rp})\ln(v_{rp})\  ,  
\end{align} 
and only using the property 
\be\label{eq:prop}  T_rv_{pq}=T_qv_{pr}+T_p v_{rq}\  , \ee  
by direct computation we find that \textit{off-shell}, ie. as an identity 
in the $u$-variables, 
\begin{align*}
(\square\mathsf{L})_{pqrs} = &  
(T_qT_sv_{pr})T_s\ln\left(\frac{v_{pq}\,T_qv_{rp}}{v_{rp}\,T_rv_{pq}}\right) 
+(T_pT_sv_{rq})T_s\ln\left(\frac{v_{pq}\,T_pv_{qr}}{v_{qr}\,T_rv_{pq}}\right) \\ 
& + (T_pT_rv_{qs})T_r\ln\left(\frac{v_{pq}\,T_pv_{qs}}{v_{qs}\,T_sv_{pq}}\right) 
+(T_qT_rv_{sp})T_r\ln\left(\frac{v_{pq}\,T_qv_{sp}}{v_{sp}\,T_sv_{pq}}\right) \\ 
& + (T_pT_qv_{rs})T_p\ln\left(\frac{v_{qr}\,T_qv_{rs}}{v_{rs}\,T_sv_{qr}}\right) 
+(T_pT_rv_{sq})T_p\ln\left(\frac{v_{qr}\,T_rv_{sq}}{v_{sq}\,T_sv_{qr}}\right) \\ 
& + (T_qT_rv_{qs})T_q\ln\left(\frac{v_{pr}\,T_rv_{sp}}{v_{sp}\,T_sv_{pr}}\right) 
+(T_pT_qv_{sr})T_q\ln\left(\frac{v_{pr}\,T_pv_{rs}}{v_{rs}\,T_sv_{pr}}\right) \\  
& + (T_qv_{ps})\ln\left(\frac{v_{pq}\,T_qv_{sp}}{v_{sp}\,T_sv_{pq}}\right) 
+(T_pv_{sq})\ln\left(\frac{v_{pq}\,T_pv_{sq}}{v_{qs}\,T_sv_{pq}}\right) \\ 
& + (T_qv_{rs})\ln\left(\frac{v_{rs}\,T_sv_{qr}}{v_{qr}\,T_qv_{rs}}\right) 
+(T_rv_{sq})\ln\left(\frac{v_{sq}\,T_sv_{qr}}{v_{qr}\,T_rv_{sq}}\right) \\ 
& + (T_rv_{qp})\ln\left(\frac{v_{pq}\,T_pv_{qr}}{v_{qr}\,T_rv_{pq}}\right) 
+(T_qv_{pr})\ln\left(\frac{v_{rp}\,T_pv_{qr}}{v_{qr}\,T_qv_{pr}}\right) \\ 
& + (T_sv_{rp})\ln\left(\frac{v_{pr}\,T_pv_{rs}}{v_{rs}\,T_sv_{rp}}\right) 
+(T_rv_{ps})\ln\left(\frac{v_{sp}\,T_pv_{rs}}{v_{rs}\,T_rv_{sp}}\right)\ ,  
\end{align*} 
where we can rewrite these terms using the KP-octahedron quantity\footnote{Octahedral 
objects like $\Omega_{pqr}$ were introduced in connection with the quadrilateral 
lattice equations in the ABS list, cf. \cite{BollPetrSur}, cf. also the more recent treatment 
\cite{RichVerm}.}  
\be\label{eq:KPOct} 
\Omega_{pqr}:= (T_rv_{pq})v_{qr}-v_{pq}T_pv_{qr}= (p-q+T_rT_qu-T_rT_pu)(T_ru-r)+{\rm cycl.} \ , 
\ee
(which obviously vanishes on solutions of \eqref{eq:potDDDKP}, but here we define it as 
a polynomial in the $u$-variables and its shifts) by writing the expressions within 
the logarithms as 
\[ \frac{(T_rv_{pq})v_{qr}}{v_{pq}T_pv_{qr}}=1+ \frac{\Omega_{pqr}}{v_{pq}T_pv_{qr}}\ , \] 
and expanding the logarithms, we find the double-zero expansion
\begin{align}
(\square\mathsf{L})_{pqrs} = & 
(T_s\Omega_{rpq})T_s\left( \frac{T_qv_{pr}}{v_{rp} T_rv_{pq}}
     -\frac{T_pv_{rq}}{v_{rq}T_rv_{qp}}+ \mathcal{O}(T_s\Omega_{rpq}) \right) \nn \\  
& +(T_r\Omega_{sqp})T_r\left( \frac{T_pv_{qs}}{v_{sq} T_sv_{qp}}
     -\frac{T_qv_{sp}}{v_{sp}T_sv_{pq}}+ \mathcal{O}(T_r\Omega_{sqp}) \right) \nn \\ 
& +(T_p\Omega_{srq})T_p\left( \frac{T_qv_{rs}}{v_{sr} T_sv_{rq}}
     -\frac{T_rv_{sq}}{v_{sq}T_sv_{qr}}+ \mathcal{O}(T_p\Omega_{srq}) \right) \nn \\
& +(T_q\Omega_{spr})T_q\left( \frac{T_rv_{ps}}{v_{sp} T_sv_{pr}}
     -\frac{T_pv_{sr}}{v_{sr}T_sv_{rp}}+ \mathcal{O}(T_q\Omega_{spr}) \right)\ ,  \label{eq:dzeroexp} 
\end{align} 
where, using \eqref{eq:prop} all the terms of lower order in the shifts cancel. 
The factors in each term of \eqref{eq:dzeroexp} vanish on solutions of the lattice 
potential KP equation $\Omega_{\cdot,\cdot,\cdot}=0$. In fact the 
second factor in each term the contributions of first and higher order 
in $\Omega$ vanish on solutions of the potential KP, while 
the terms of the form 
\[ \frac{T_qv_{pr}}{v_{rp} T_rv_{pq}}-\frac{T_pv_{rq}}{v_{rq}T_rv_{qp}} \] 
in the second factors equally vanish on solutions of the lattice 
KP equation.  
\end{proof}

When considering the generalised EL equations 
$\delta (\square\mathsf{L})_{pqrs}=0$ one has to take into account the identity 
\[ T_s\Omega_{pqr}-T_p\Omega_{qrs}+T_q\Omega_{rsp}-T_r\Omega_{spq}=0 \ , \] 
as an identity in the variables $u$ and its shifts, and thus the different 
prefactors of the form $T\Omega$ 
in the double-zero expansion are not all independent. This leads to a slightly 
weaker form of the EL equations coming from the computation   
\[ \delta\square\mathsf{L}= \sum \left[ (\delta T\Omega)\left( \frac{Tv}{vTv}+\cdots \right)+(T\Omega)\delta\left( \frac{Tv}{vTv}+\cdots \right)\right]=0\  ,   \] 
while splitting the contributions up into functionally independent components. 
In conclusion, the double-zero structure provides the relevant EL equations for the 
potential KP system, and thus establishes the multiform structure for the fully discrete KP system.

 \section{Discussion} 
 
One of the initial motivations of this paper was to complete the known link between the Calogero-
Moser systems and the KP system, a link first discovered by Krichever for the conventional 
CM system and the standard KP equation, \cite{Krich}, and later 
extended to the discrete-time case in \cite{NijPeng} based on a pole-reduction of the semi-discrete KP equation  
\eqref{eq:DddKP}. The corresponding Lagrangian 1-form 
structure was explored in \cite{Y-KLN}, using the Calogero-Moser system as a laboratory to 
establish the EL structure of the corresponding Lagrangian 1-forms, but this was 
partly hampered by the absence of a 
Lagrangian structure for \eqref{eq:DddKP}, which disallowed us to perform the pole-reduction on the Lagrangian level. That lacuna is hopefully now repaired in the present paper, but the actual pole-reduction from 3-forms to 1-forms remains to be achieved. 

More generally, the reduction problem from 3-forms to, say, 2-forms is something that 
needs to be explored, not only for reductions from the KP system to the corresponding 
2D quad equation (H1 of the ABS list), 
but the more subtle and highly nontrivial reduction to higher-order 2D lattice 
systems like the lattice Boussinesq system (cf. e.g. \cite{NijZhang23}) or the 
lattice Gel'fand-Dikii hierarchy. 

Another issue that emerges is the further study of the generating PDE 
system \eqref{eq:uvPDE} for the KP hierarchy, in particular finding the 
Lagrange structure.  In the linearised case the corresponding PDE can be 
written entirely in terms of $u$ alone, and reads 
\begin{align}\label{eq:linKPPDE} 
 (p-q)^2 u_{\xi\xi\sigma\sigma\tau\tau}-(p-q)^4 u_{\sigma\sigma\tau\tau} + 
 2(p-q)^2 \left(u_{\sigma\tau\tau\xi}+u_{\sigma\sigma\tau\xi} \right) 
 = u_{\sigma\sigma\xi\xi}+u_{\tau\tau\xi\xi}-2u_{\sigma\tau\xi\xi}\  ,  
\end{align} 
which possesses a Lagrangian 
\begin{align}
\mathcal{L}= \tfrac{1}{2}(p-q)^2 u_{\sigma\tau\xi}^2+\tfrac{1}{2}(p-q)^4 u_{\sigma\tau}^2 
-(p-q)^2 u_{\sigma\tau}\left(u_{\tau\xi}+u_{\sigma\xi} \right) 
 + \tfrac{1}{2}\left(u_{\sigma\xi}^2+u_{\tau\xi}^2-2u_{\sigma\xi}u_{\tau\xi} \right)\ . 
\end{align} 
How the latter Lagrangian fits into a 3-form structure remains for now an open problem. Ideally, 
one would like to find another generating PDE system in which the special variable 
$\xi$ is replaced by yet another Miwa variable, say $\rho$, which would complete the picture. Such a PDE system, in terms only of Miwa variables, does exist, but not for the 
dependent variable $u$; it is the Darboux-KP system the Lagrangian 3-form structure of 
which was found in \cite{Nij2023}. Intriguingly that Lagrangian structure is in fact a 
curious kind of infinite-dimensional Chern-Simons theory, as was found in \cite{FMNR2023}. 
The latter finding begs for a further exploration of the mysterious connection between 
topological field theories and Lagrangian multiform systems.

\subsection*{Acknowledgements}

The author has benefited from useful discussions with V. Caudrelier, L. Peng, J. 
Richardson,  D. Sleigh, M. Vermeeren and special thanks to C. Zhang and D-J. Zhang 
for their support. He is grateful for the 
hospitality at the Mathematics Department of Shanghai University where this work 
partly written. The work was partly supported by the Foreign Expert Program of the Ministry of Sciences and Technology of China, grant no: G2023013065L.

\label{lastpage}
\end{document}